\documentclass[11pt,a4paper]{article}

\usepackage{mathtools}
\usepackage{authblk} % author
\usepackage{algpseudocode,algorithmicx,algorithm}

\usepackage{mathrsfs}
\usepackage{latexsym,bm}
\usepackage{amsmath,amsfonts,amsmath,amssymb,amsthm}
\usepackage{extarrows}

\usepackage{graphicx,subfigure,epstopdf,float}
\usepackage{enumerate,cases,multirow}
\usepackage{makecell}
\usepackage{caption}

\usepackage{longtable,colortbl,arydshln,threeparttable}
\definecolor{mygray}{gray}{.9}

\usepackage{indentfirst}
\usepackage[top=25mm,bottom=20mm,left=25mm,right=20mm]{geometry}
\usepackage{cite}

\usepackage{listings}

\usepackage{makeidx}        % generate an index, automatically
\usepackage{booktabs}
\usepackage[bookmarksnumbered,colorlinks,citecolor=red,linkcolor=red,hyperindex,linktocpage=true]{hyperref}

\newcommand{\ket}[1]{| #1 \rangle} % |u>
\newcommand{\bb}{\boldsymbol}

\def \d {\mathrm{d}}
\def \e {\mathrm{e}}
\def \i {\mathrm{i}}

\newcounter{parentalgorithm}

\makeatother

\newtheorem{theorem}{Theorem}[section]
\newtheorem{lemma}{Lemma}[section]

\newtheorem{definition}{Definition}[section]
\newtheorem{example}{\bf Example}[section]
\theoremstyle{remark}
\newtheorem{remark}{\bf Remark}[section]

\numberwithin{equation}{section}

\begin{document}

\title{Quantum simulation of partial differential equations via Schr\"odingerisation: technical details}
\author[1,2]{Shi Jin \thanks{shijin-m@sjtu.edu.cn}}
\author[1, 2, 3]{Nana Liu\thanks{nana.liu@quantumlah.org}}
\author[1]{Yue Yu\footnote{Corresponding author.}\thanks{terenceyuyue@sjtu.edu.cn}}
\affil[1]{School of Mathematical Sciences, Institute of Natural Sciences, MOE-LSC, Shanghai Jiao Tong University, Shanghai, 200240, P. R. China.}
%\affil[2]{Institute of Natural Sciences, Shanghai Jiao Tong University, Shanghai 200240, China.}
%\affil[3]{Ministry of Education, Key Laboratory in Scientific and Engineering Computing, Shanghai Jiao Tong University,
%Shanghai 200240, China}
\affil[2]{Shanghai Artificial Intelligence Laboratory, Shanghai, China.}
\affil[3]{University of Michigan-Shanghai Jiao Tong University Joint Institute, Shanghai 200240, China.}

%\date{}

\maketitle

\begin{abstract}
  We study a new method-- called {\it Schr\"odingerisation} introduced in \cite{Schrshort} -- for solving general linear partial differential equations with quantum simulation. This method converts linear partial differential equations into a `Schr\"odingerised' or Hamiltonian system, using a new and simple transformation called the {\it warped phase transformation}. Here we provide more in-depth technical discussions and expand on this approach in a more detailed and pedagogical way. We apply this to more examples of partial differential equations, including heat, convection, Fokker-Planck, linear Boltzmann and Black-Scholes equations. This approach can also be extended to Schr\"odingerise general linear partial differential equations, including the Vlasov-Fokker-Planck equation and the Liouville representation equation for nonlinear ordinary differential equations.
  %We also show this new approach finds a variety of applications in time-dependent boundary value problems.
\end{abstract}


\tableofcontents

\section{Introduction}

Quantum algorithms for solving partial differential equations (PDEs) have received extensive attention in recent years \cite{Cao2013Poisson,Berry-2014,qFEM-2016,Costa2019Wave,Engel2019qVlasov,Childs-Liu-2020,Linden2020heat,
Childs2021high,JinLiu2022nonlinear,GJL2022QuantumUQ,JLY2022multiscale} thanks to the development of quantum algorithms with exponential acceleration advantages in linear algebraic problems \cite{HHL2009,Childs2017QLSA,Costa2021QLSA,Berry-2014,BerryChilds2017ODE,Childs-Liu-2020,Subasi2019AQC}. For time-dependent PDEs, one usually discretises the spatial variables to get a system of ordinary differential equations (ODEs), which in turn is solved by quantum ODE solvers \cite{Berry-2014,BerryChilds2017ODE,Childs-Liu-2020}. In particular, when the resulting ODE is also a Hamiltonian system, quantum simulation methods can be performed. In general, quantum simulations have less time complexity than quantum ODE solvers or other quantum linear algebra solvers (e.g., the quantum difference methods \cite{Berry-2014,JLY2022multiscale}), and thus the design for quantum simulation algorithms for solving linear PDEs is important for a wide range of applications. A very recent proposal appeared based on block-encoding \cite{An2022blockEncodingODE}.

This paper presents more in-depth technical details for a new protocol that transforms a general linear PDE into a quantum Hamiltonian system. We call this new method the {\it Schr\"odingerisation method}, introduced in our short companion paper \cite{Schrshort}. From the simplest example of the heat equation, for instance, a corresponding set of Schr\"odinger's equations can be derived. We present the heat equation example in detail, as a warm-up example to familiarise with the technique, in Section~\ref{sect:Schrodingerisation}. We also present the example of the convection equation.

This method is inspired by a recently developed transformation given in \cite{GJL2022QuantumUQ}, though originally for a completely different motivation. Here we introduce an auxiliary variable and construct transformations~-~referred to as the {\it warped phase transformation}~-~that converts the original equation into a new equation that has the structure of the Schr\"odinger operator. Then it can subsequently be simulated by quantum Hamiltonian simulation. Since the method introduces only a one-dimensional auxiliary variable, the additional computational cost is small.

When discretising a general linear PDE by spectral methods or other numerical methods, a Hamiltonian system is not necessarily obtained. There are two reasons for this, one is that the coefficient matrix arising from each term of the equation is not always preceded by the imaginary unit $\i = \sqrt{-1}$, and the other is that despite the imaginary unit, the coefficient matrix is usually not symmetric, especially for problems with variable coefficients, or for the discretisation schemes that are not centered (with a symmetric stencil).

We observe that applying the warped phase transformation works for equations with constant coefficients, and it also works for some variable coefficient problems. However, for general variable coefficient problems, the direct use of the method does not necessarily achieve the goal of getting a Hamiltonian system. For example, see the Vlasov-Fokker-Planck equation discussed in Subsect.~\ref{subsec:VFP}.

For this reason, we further design a universal algorithm for the system of linear ODEs based on warped transformation,  where the system of ODEs is obtained after spatial discretisations of any linear PDEs. In Section~\ref{sec:general} we introduce an algorithm for this general linear system of ODEs. This idea works ODEs resulting from spatial discretizations for all constant coefficient and even some variable coefficients PDEs.  We also discuss some alternative methods to our Schr\"odignerisation approach.

We remark that this approach will also find a variety of applications in solving problems with a (time-dependent or independent) source term, and boundary value problems (with the time-dependent or independent boundary conditions).  In fact, when constructing Hamiltonian systems for general boundary value problems (for example with the Dirichlet boundary condition), the inhomogeneous right-hand side in the resulting ODE system may arise after spatial discretization of the boundary conditions for the PDEs.  We propose a simple augmentation technique (see Remark~\ref{rem:augment}) to resolve this issue, which together with the warped phase transformation gives the Schr\"odingerisation approach for time-dependent boundary value problems.

%Don't need to discuss here in the into: There is a different protocol~-~the unitarisation method~-~presented for the Black-Scholes equation in \cite{Javier2022optionprice}. We extend its protocol to general cases and compare it with our  Schr\"odingerisation approach. In comparison, our method retains the structure of the original problem without the additional approximation of the operators involved in the unitary dilation operators. This means our protocol is easier to implement. {\color{red}However, our approach has additional computational cost due to the CFL condition when the arccos problem can be easily resolved for the unitarisation method as discussed in Sect.~\ref{sec:general}}.
In Section~\ref{sec:applications}, we show that our method is applicable to a variety of important linear partial differential equations, including the linear heat, convection,  (Vlasov-) Fokker-Planck, linear Boltzmann and Black-Scholes equations. For nonlinear problems, we give an application via the linear representation~-~the Liouville representation~-~approach for nonlinear dynamical systems. It is worth pointing out that the Liouville representation can be symmetrised using the Koopman-von Neumann representation \cite{Joseph2020KvN,Dobin2021plasma,JinLiu2022nonlinear,JinLiuYu2022nonlinear}, but it involves the square root of the Dirac delta-function $\delta (x)$, which is not well-defined  mathematically, even in the weak sense. Thus one needs to be more careful in interpreting its solution and the consequent numerical convergence in a suitable solution space \cite{JinLiuYu2022nonlinear}. Our Schr\"odingerisation approach allows a direct treatment of the Liouville equation, and hence does not have difficulties in this regard.

%Overall we sort of `Schr\"odingerise' the linear PDEs so they fit quantum computer naturally.

\section{Schr\"odingerisation of the heat and convection equations} \label{sect:Schrodingerisation}

\subsection{Schr\"odingerisation of the heat equation}

This section demonstrates how to transform a linear heat equation into Schr\"odinger-type PDEs.

\subsubsection{The reformulation of the heat equation} \label{subsect:idea}

Consider the following initial value problem of the linear heat equation
\begin{equation*} \label{heateq0}
\begin{cases}
\partial_tu - \Delta u = 0, \\
u(0,x) = u_0(x),
\end{cases}
\end{equation*}
where $u = u(t,x)$, $x = (x_1,x_2,\cdots,x_d) \in \mathbb{R}^d$ is the position and $t\ge 0$.
%Without loss of generality, we assume $x_i \in [-1,1]$ for $1\le i \le d$ and impose the periodic boundary conditions throughout this article.

Introduce an auxiliary variable $p>0$ and define
\[w(t,x,p) = \e^{-p} u(t,x), \qquad p>0.\]
A simple calculation shows that $w$ solves
\begin{equation}\label{heatreformulation}
\partial_t w + \partial_p \Delta_x w = 0,  \qquad p>0.\\
\end{equation}
From $w$ one can recover $u$ via
\begin{equation}\label{integration}
  u(t,x)= \int_0^\infty w(t,x,p)\d p = \int_{-\infty}^\infty \chi(p) w(t,x,p) \d p,
  \end{equation}
where $\chi(p) = 1$ for $p>0$ and $\chi(p) = 0$ for $p<0$, or, since $u(t,x)=\e^p w(t,x,p)$ for all $p>0$, one can simply choose any $p_*>0$ and let
\begin{equation} \label{point}
  u(t,x) = \e^{p_*} w(t,x, p_*).
\end{equation}

Applying the Fourier transform on $x$ (here we assume  $x$ is defined in a periodic domain) and let $\hat{w}(t,\xi,p)$,  with $\xi = [\xi_1,\cdots,\xi_d]^T$ being the Fourier modes, be the  corresponding Fourier transform of $w$, one gets a convection equation
\[\partial_t \hat{w}- |\xi|^2 \partial_p \hat{w} = 0,\]
where $|\xi|^2 = \xi_1^2 + \cdots + \xi_d^2$. Clearly the solution $\hat{w}$ moves from the right to the left, so no boundary condition is needed at $p = 0$, while the initial data of $w$ is given by
\[
w(0,x,p)= \e^{-p} u_0(x), \qquad p>0.
\]

Moreover, if we extend $w$ to $p<0$, then the solution does not impact the region $p>0$ for $w$. For this reason, we symmetrically extend the initial data of $w$ to $p<0$ but keep Eq.~\eqref{heatreformulation}:
\begin{equation}\label{heatreformulationextend}
\begin{cases}
\partial_t w +  \Delta_x \partial_p w= 0, \qquad p \in (-\infty, \infty), \\
w(0,x,p) = \e^{-|p|} u_0(x).
\end{cases}
\end{equation}
This equation for $w$ will be called {\it phase space heat equation}. The solution  obviously coincides with the solution of \eqref{heatreformulation} when $p>0$. Due to the exponential decay of $\e^{-|p|}$ one can (computationally) impose the periodic boundary condition $w(t,x,p=-L) = w(t,x,p = L)$ $( = 0)$ along the $p$-direction for some $L>0$ sufficiently large. Then the Fourier transform on
$p$ gives
\begin{equation}\label{heat-Schro}
  \partial_t\tilde{w} -\i \eta \Delta\tilde{w}=0 \qquad \mbox{or} \qquad \i \partial_t\tilde{w} = - \eta \Delta\tilde{w},
\end{equation}
where $\tilde w(t,x,\eta)$, $\eta\in \mathbb{R}$, is the Fourier transform of $w$ in $p$.  Equation \eqref{heat-Schro} is clearly the Schr\"odinger equation, for every $\eta$!

\begin{remark}\label{rem:idea}
 One can justify the validity of the set up \eqref{heatreformulationextend}~--~which is more convenient than the half space problem for $w$ if one uses, for example spectral methods in $p$ space~--~in another way. Applying the Fourier transform on $x$ to \eqref{heatreformulationextend}, one gets
 \[\begin{cases}
\partial_t \hat{w}_t(t,\xi, p) - |\xi|^2 \partial_p \hat{w}(t,\xi, p) = 0, \qquad p \in (-\infty, \infty), \\
\hat{w}(0, \xi, p) = \e^{-|p|} \hat{u}_0(\xi).
\end{cases}\]
Using the method of characteristics, the analytic solution is given by
\begin{equation}\label{what-eqn}
\hat{w}(t,\xi, p) = \hat{w}(0, p+|\xi|^2 t) = \e^{-|p + |\xi|^2 t|} \hat{u}_0(\xi).
\end{equation}

 If $p>0$, then
\[\hat{w}(t,\xi,p) = \e^{-|p + |\xi|^2 t|} \hat{u}_0 = \e^{-p} \Big( \e^{-|\xi|^2 t} \hat{u}_0(\xi) \Big).\]
  The inverse transform implies
 \[w(t,p) = \e^{-p} \mathcal{F}^{-1} \Big( \e^{-|\xi|^2 t} \hat{u}_0(\xi) \Big) = \e^{-p} u(t,x),\]
  where $\mathcal{F}$ represents the Fourier transform. This  is exactly the  solution of \eqref {heatreformulation}.

From \eqref{what-eqn}, one sees that $|\hat{w}(t,\xi,p)|\le |\hat{u}_0(\xi)|$. Note that if $u_0(x) \in C^k$, then $\hat{u}_0(\xi) \sim \mathcal{O}(1/|\xi|^k)$. Therefore if
$u_0(x)$ is sufficiently smooth, $\hat{w}$ decays very fast in $\xi$, which means only those $\xi$ such that $|\xi|=\mathcal{O}(1)$
are important, hence $\hat{w}$ will move to the left with $\mathcal{O}(1)$ speed. Thus for $T=\mathcal{O}(1)$, $|L|=\mathcal{O}(1)$ is sufficient for the computational domain of $p\in [L, R]$, where $L<0$ and $R>0$.
\end{remark}

 From Remark \ref{rem:idea} and Eq.~\eqref{heatreformulationextend}
one sees easily that
\[
\int_{-\infty}^\infty \int_0^\infty w(t,x,p)^2 \d p \d x
=\frac{1}{2}\int_{-\infty}^\infty u(t,x)^2 \d p \d x,\]
\[
\int_{-\infty}^\infty \int_{\infty}^\infty w(t,x,p)^2 \d p \d x
=\int_{-\infty}^\infty u_0(x)^2 \d x.
\]
Standard PDE analysis using Poincare's inequality gives
\begin{equation}\label{U-W-normalization}
\int_{-\infty}^\infty \int_0^\infty w(t,x,p)^2 \d p \d x
\le \e^{-\alpha t} \int_{-\infty}^\infty \int_{\infty}^\infty w(t,x,p)^2 \d p \d x,
\end{equation}
for some positive constant $\alpha$ that depends on the (finite) domain size in $x$.

%{\color {green} Suppose the computational domain of $w$ for $p$ is $p\in -[L, L]$. We propose another way to obtain $u$ from $w$:
%\begin{equation}\label{w-to-u-3}
%u(t,x)=\frac{1}{L}\int_0^L e^pw(t,x,p)\, dp
\ .
%\end{equation}
%As will be shown later, compared with \eqref{integration} and \eqref{point}, this will significantly increase the probability of getting $u$ from the quantum state of $w$ computed from the quantum algorithm, thus reduces the overall computational cost.
%}

\begin{example}
We conduct a numerical test in one dimension to justify the above idea:
\[
\begin{cases}
u_t - u_{xx} = 0,  \qquad x \in (-1, 1),\\
u(0,x) = u_0(x),\\
u(t,-1) = u(t,1), \quad u_x(t,-1) = u_x(t,1).
\end{cases}
\]
The exact solution is given by $u(t,x) = \e^{-\pi^2 t} \sin(\pi x)$.
\end{example}

For numerical implementation, it is natural and convenient to introduce  $\alpha = \alpha(p)$ in the initial data of \eqref{heatreformulationextend} for $p<0$:
\begin{equation}\label{heatex}
\begin{cases}
\partial_t w + \partial_{xxp} w= 0, \qquad p \in (-\infty, \infty), \\
w(0,x,p) := w_0(x,p) = \e^{-\alpha |p|} u_0(x).
\end{cases}
\end{equation}
To match the exact solution, $\alpha(p) = 1$ is necessary for the region $p> 0$. In the $p>0$-domain, we will truncation the domain at $p=R$, where $R$ is sufficiently large such that $\e^{-R} \approx 0$. We will choose a large $\alpha$ for $p<0$ so the solution (see Fig.~\ref{fig:domain}) will have a support within a relatively small domain.  Since  the wave $\hat{w}$ moves to the left, one needs to choose the artificial boundary at $p=L<0$, for $|L|$ large enough such that $\hat{w}$, initially almost compact at $[L_0, R]$, will not reach the point $p=L$ during the duration of the computation. This will allow  to
use periodic boundary condition in $p$ for spectral approximation.

\begin{figure}[H]
 \centering
  \includegraphics[height=5cm,width=10cm]{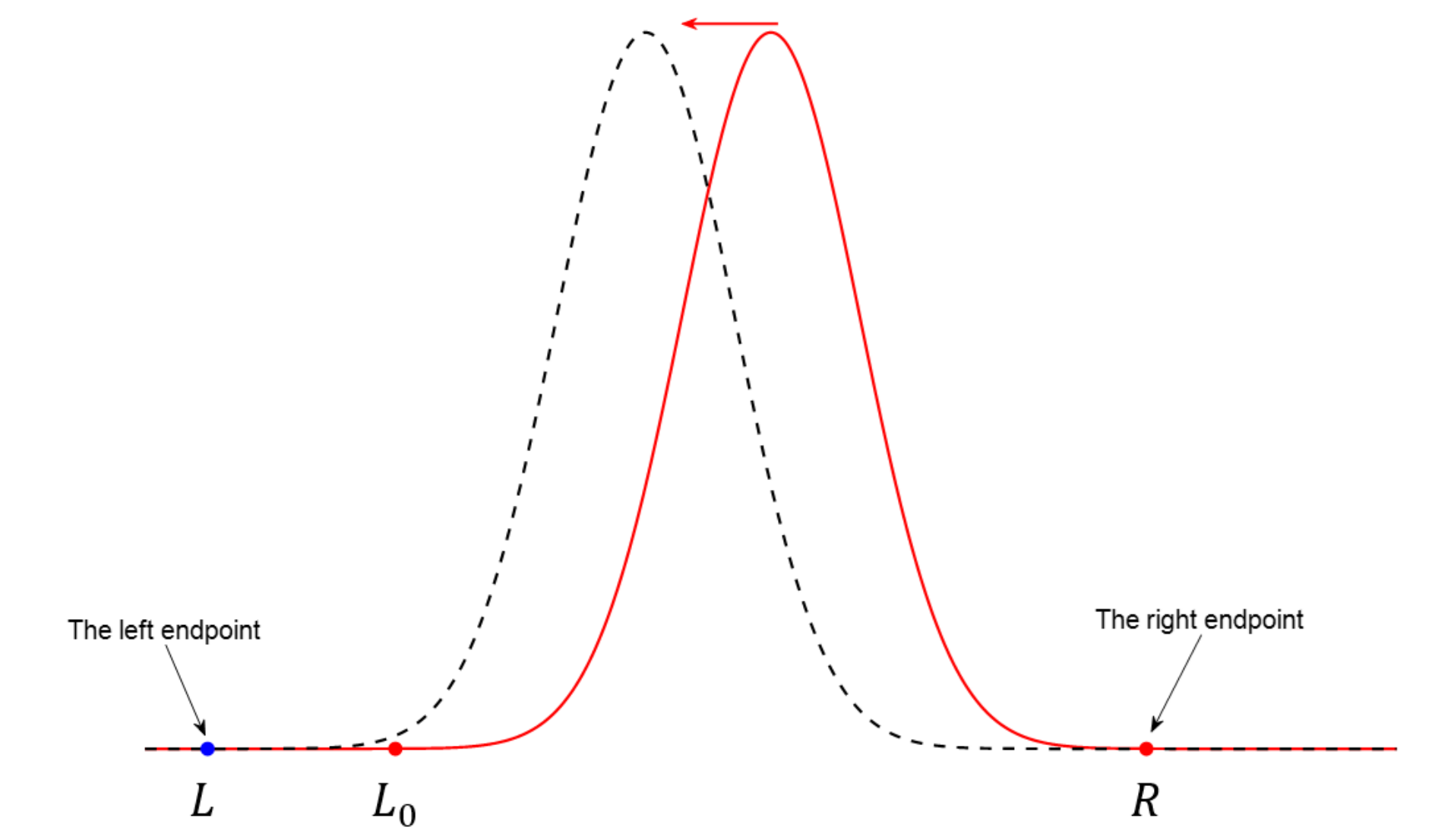}\\
  \caption{Schematic diagram for the computational domain of $p$}\label{fig:domain}
\end{figure}

The Fourier spectral approach will be used to discretise the spatial and the auxiliary variables. Let $M$ and $N$ be two even numbers.
We choose uniform mesh sizes $\Delta x = 2/M$ and $\Delta p = (R-L)/N$ for the spatial and the auxiliary variables, with the grid points denoted by $x_0<x_1<\cdots<x_M$ and $p_0<p_1<\cdots<p_N$, respectively. Let $\bb{w}(t,p) = [w(t,x_0,p), w(t,x_1,p), \cdots, w(t,x_{M-1}, p)]^T$. The discrete Fourier transform (DFT) on $x$ gives
\begin{equation}\label{originalw}
\begin{cases}
\partial_t \bb{w}(t,p) -  P_\mu^2 \partial_p\bb{w}(t,p) = 0, \qquad p \in (L, R), \\
\bb{w}(0,p) = \e^{-\alpha |p|} \bb{u}_0,
\end{cases}
\end{equation}
where $\bb{u}_0 = [u(0,x_1), \cdots, u(0,x_{M-1})]^T$, and $P_\mu$ is the matrix representation of the momentum operator $-\i \partial_x$ in the original variables. For details on the derivation of \eqref{originalw}, please refer to the notations in the next subsection. The matrix $P_\mu$ can be diagonalised via $D_\mu = \Phi^{-1} P_\mu \Phi$, where $D_\mu = \text{diag}(\mu_{-M/2}, \cdots, \mu_{M/2-1})$ is a diagonal matrix with $\mu_l = \pi l$ for $l=-M/2,\cdots, M/2-1$. Let $\hat{\bb{w}}(t,p) = \Phi^{-1} \bb{w}(t,p)$. Then one has
\begin{equation}\label{wavewhat}
\begin{cases}
\partial_t \hat{\bb{w}}(t,p) -  D_\mu^2 \partial_p\hat{\bb{w}}(t,p) = 0, \qquad p \in (L, R), \\
\hat{\bb{w}}(0,p) = \e^{-\alpha |p|} \hat{\bb{u}}_0,
\end{cases}
\end{equation}
where $\hat{\bb{u}}_0 = \Phi^{-1} \bb{u}_0$, and the $l$-th component of $\hat{\bb{w}}$, denoted by $\hat{\bb{w}}_l$,  corresponds to a linear hyperbolic system  and the wave moves from the right to the left with speed $s_l = \mu_l^2$. The analytic solution to \eqref{wavewhat} is obviously given by
\begin{equation}\label{analyticwhat}
\hat{\bb{w}}_l(t,p) = \e^{-\alpha (p+s_l t) |p+s_l t|} \hat{\bb{u}}_{0,l}, \qquad l = -M/2,\cdots, M/2-1.
\end{equation}
 According to the discussion in Remark \ref{rem:idea}, when $u_0(x)$ is sufficiently smooth, we know that the fastest left moving wave will have a speed $s_*= \mathcal{O}(1)$. Given the evolution time $T$, we can estimate a large enough $|L|$ such that $ |s_* T| < |L|$.

Now we consider the full discretisation. Let the grid value matrix $W(t):=(w(t,x_i,p_j))_{M\times N}$, which can be straightened as
\[\bb{w}(t) = [\bb{w}_0; \bb{w}_1; \cdots; \bb{w}_{M-1}] = \sum\limits_{ij} w(t,x_i,p_j) \ket{i,j},\]
where  ``;" indicates the straightening of $\{\bb{w}_i\}_{i\ge 1}$ into a column vector and
\[
\bb{w}_i = [w(t,x_i,p_0), w(t,x_i,p_1), \cdots, w(t,x_i,p_{N-1})]^T.
\]
Performing the DFT on both $x$ and $p$ yields
\begin{equation}\label{heatdiscrete}
\partial_t \bb{w}(t) - \i (P_\mu^2\otimes P_\mu) \bb{w}(t) = \bb{0},
\end{equation}
where we use the same notations $P_\mu$ for both variables since no confusion will arise. Let $F_x$ and $F_p$ be the discrete Fourier transform matrices for $x$ and $p$, respectively. One can translate Eq.~\eqref{heatdiscrete} into the frequency space:
\[\partial_t \bb{c}(t) - \i (D_\mu^2\otimes D_\mu) \bb{c}(t) = \bb{0},\]
where $\bb{c}(t) = (F_x^{-1} \otimes F_p^{-1}) \bb{w}(t)$. If $\bb{c}(t)$ is arranged as a matrix $C(t)$ in the form of $W(t)$, then the following relation is easily found:
\[W = Q C P^T = F_p C F_x^T \qquad \mbox{or} \qquad C = Q^{-1} W P^{-T} = F_p^{-1} W F_x^{-T},\]
which avoids the use of memory-consuming tensor products. Therefore, the numerical realisation can be effectively implemented via the discrete Fourier transform (see Remark \ref{rem:FPhi}). The solution $\hat{\bb{w}}_l(t,p)$ of \eqref{wavewhat} can be extracted from the $l$-th row of $F_x^{-1}W$.

\begin{figure}[H]
 \centering
  \includegraphics[scale=0.5]{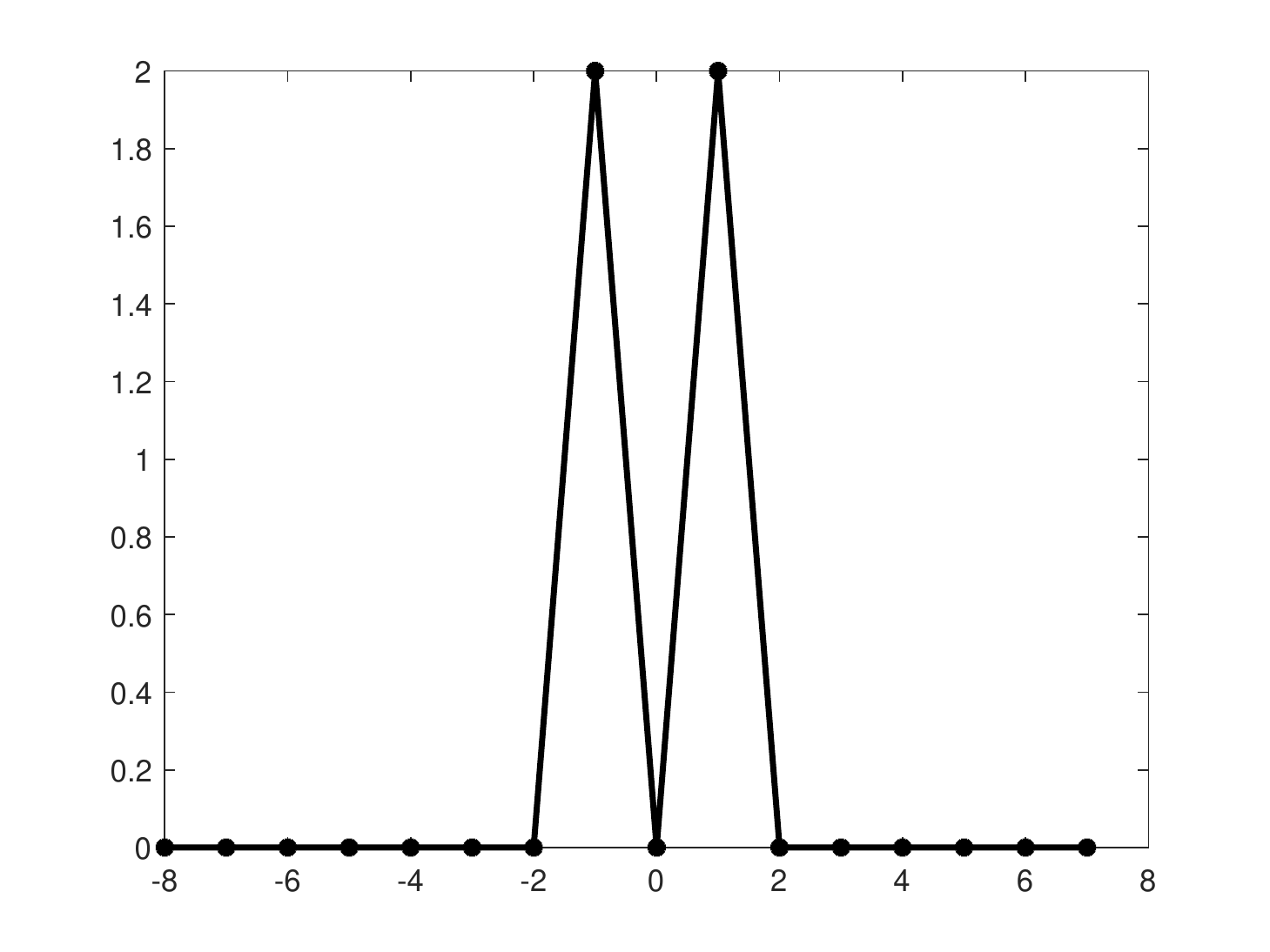}\\
  \caption{Discrete Fourier transform of the initial data $u_0(x) = \sin(\pi x)$.}\label{fig:u0hat}
\end{figure}

In the numerical test, we choose $M = 2^4 = 16$ and $N = 2^9 = 512$.   In Fig.~\ref{fig:u0hat}, we plot the modulus of the DFT coefficients $\hat{\bb{u}}_0$. Clearly, the amplitudes decay very fast in the Fourier mode $\mu_l$, and $\mu_{\pm 1}$ contribute most to the propagation. Therefore, for this example one can choose time $t$ such that
\begin{equation}\label{tchoice}
s_1 t \le L_0 - L \qquad \mbox{or} \qquad t \le T_* := (L_0-L)/s_1,
\end{equation}
where $s_1 = \pi^2 = \mathcal{O}(1)$.
For other parameters, we choose $L = -5$, $R=5$, $t = T_*$, $\alpha = 10$ for $p<0$ and $L_0 = -1$ (the estimated $T_*=0.4053$). The numerical solutions for $\bb{u}(t) = [u(t,x_0), \cdots, u(t,x_M)]^T$ are displayed in Fig.~\ref{fig:p}, with \eqref{point} and \eqref{integration} used to restore the solutions, respectively. Note that it is better to pick the point $p_*>0$ near $p=0$, for example, $p_* = p_{\frac{N}{2}+3}$, to avoid the loss of significant digits of $\e^{-p*}$ for large $p^*$.

\begin{figure}[H]
 \centering
  \subfigure[use \eqref{point}]{\includegraphics[scale=0.5]{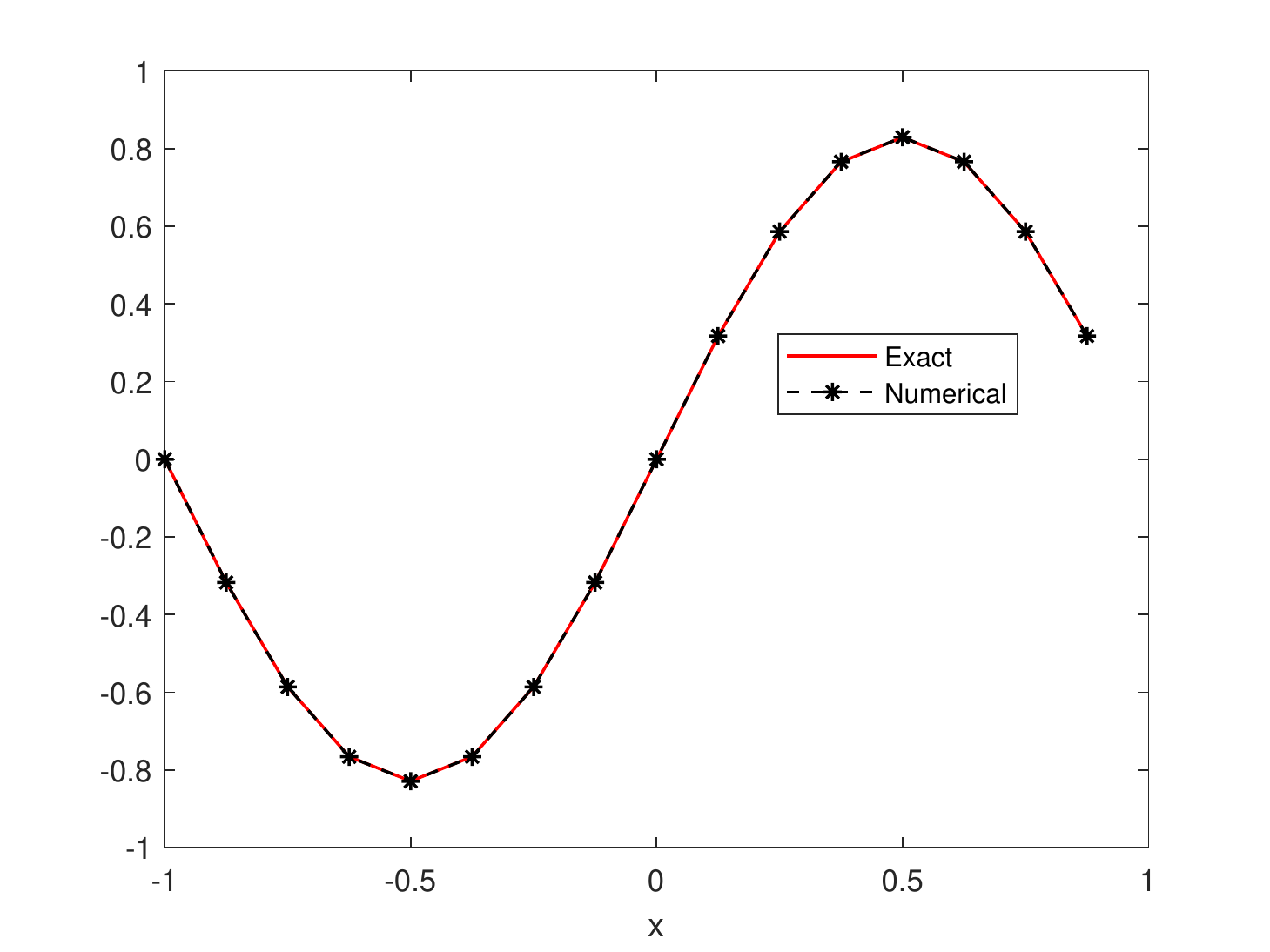}}
  \subfigure[use \eqref{integration}]{\includegraphics[scale=0.5]{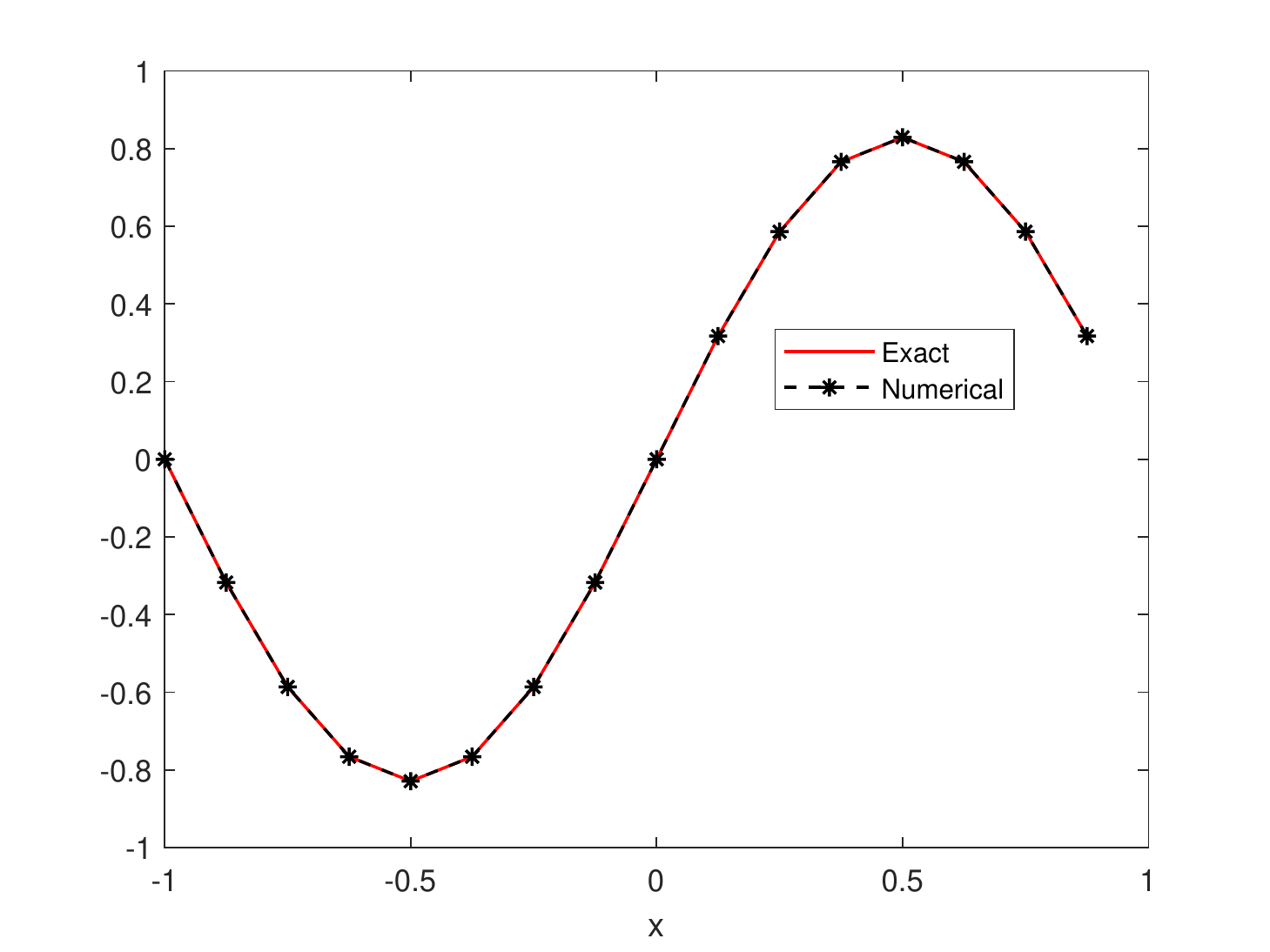}}\\
  \caption{Numerical and exact solutions of $\bb{u}(t = T_*)$ for the spectral method.}\label{fig:p}
\end{figure}

To validate the choice of computational domain  shown in Fig.~\ref{fig:domain}, we now take snapshots of  the moving of the wave corresponding to Eq.~\eqref{wavewhat}, with the result shown in Fig.~\ref{fig:what}a. Since the wave amplitude in the frequency space is a complex number and its real part is small, we in this figure display the modulus of these complex numbers, i.e., $|\hat{\bb{w}}_{l_*}(t,p_j)|$ for $j=0,1,\cdots, N-1$, where $l_*$ corresponds to the speed $s_*$. The blue line represents the initial wave $|\hat{\bb{w}}(0,p)|$, and the red and the black ones are respectively the analytic solution given by \eqref{analyticwhat} and the numerical solution at time $t = T_*$. We also plot the wave amplitudes $\bb{w}_{i_*}(t,p_j)$ in the original space in Fig.~\ref{fig:what}b, where $i_*$ corresponds to local index of $s_*$ in $\{s_i\}$. As observed, the waves in both spaces have almost moved to the left end, which validates the previous arguments.

\begin{figure}[H]
 \centering
  \subfigure[in the frequency space]{\includegraphics[scale=0.5]{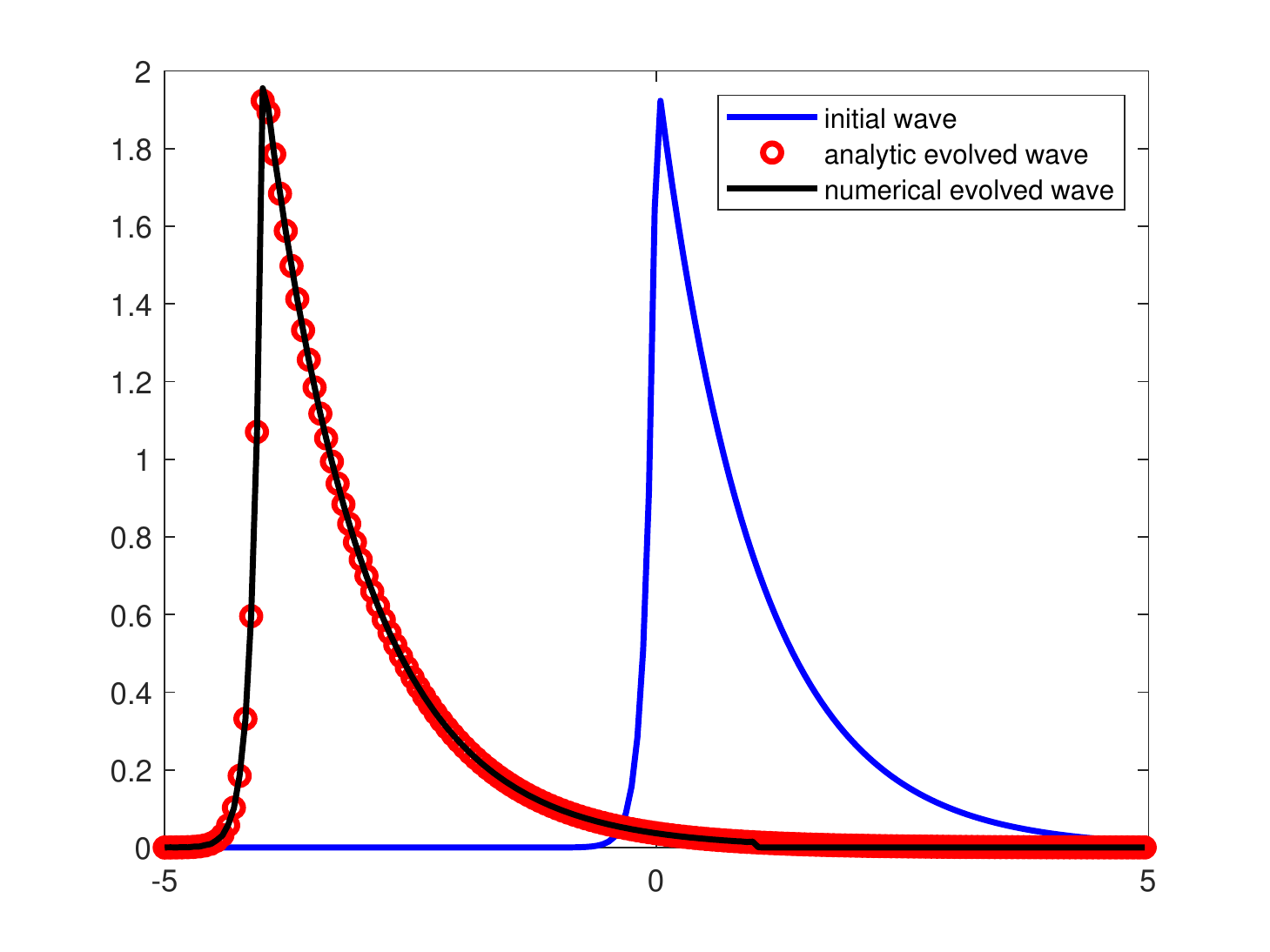}}
  \subfigure[in the original space]{\includegraphics[scale=0.5]{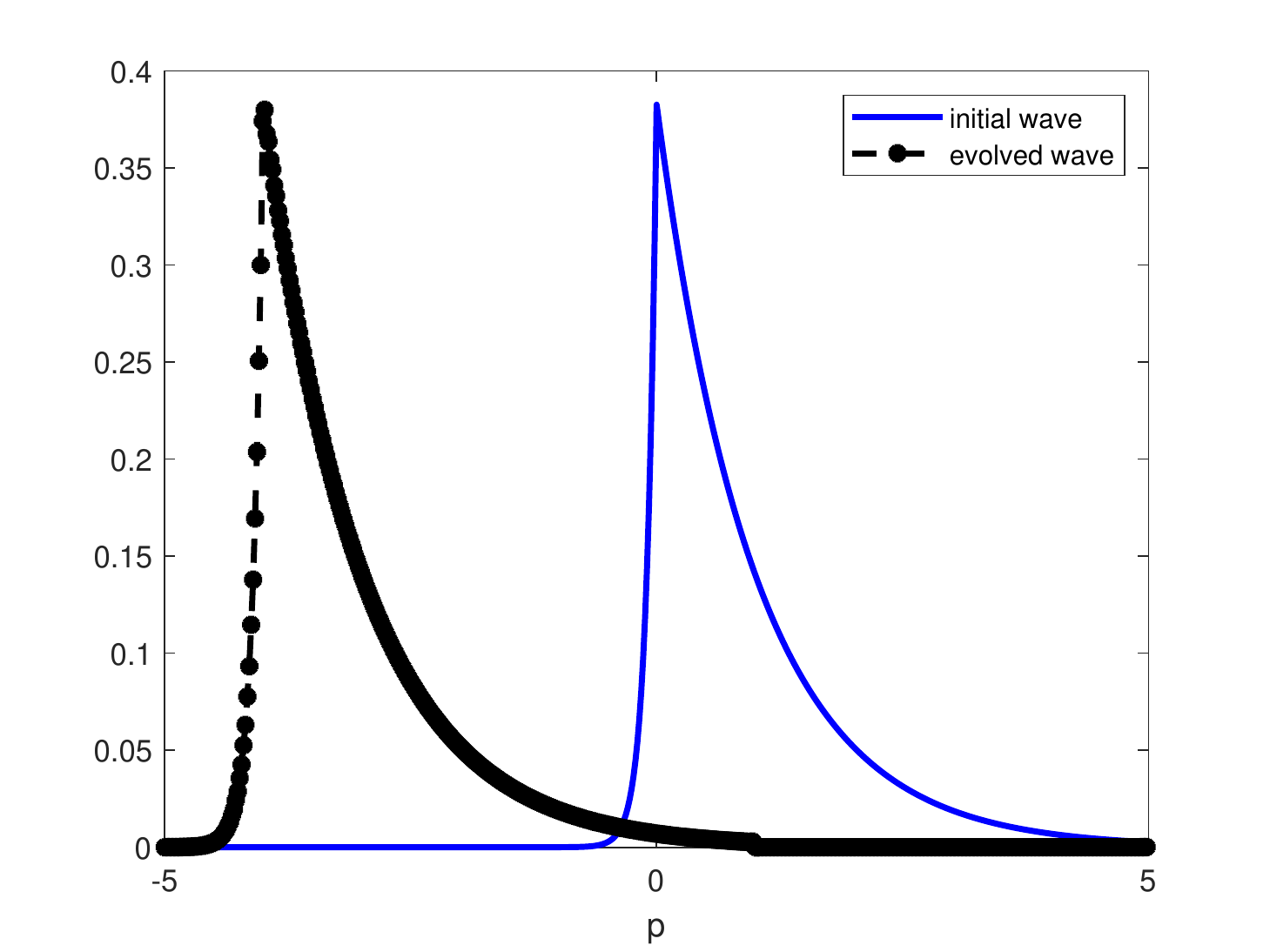}}
  \\
  \caption{The initial and evolved waves in the frequency and original spaces ($T_*=0.4053$). (a) the modulus of the wave amplitude $|\hat{\bb{w}}_{l_*}(t,p_j)|$ corresponding to \eqref{wavewhat} in the frequency space; (b) the wave amplitude $\bb{w}_{i_*}(t,p_j)$ corresponding to \eqref{originalw} in the original space.}\label{fig:what}
\end{figure}

Given $t = T = 1$, we can choose a large enough $L$ in absolute value to get satisfactory numerical results. In view of the relation \eqref{tchoice}, we take $L = L_0-T s_1 \approx 11$. Considering the periodic condition, one can choose $\alpha = 40$ and $R = 10$ for example, with the results shown in Fig.~\ref{fig:wwres}.

\begin{figure}[H]
  \centering
  \subfigure[$u(t,x)$]{\includegraphics[scale=0.35]{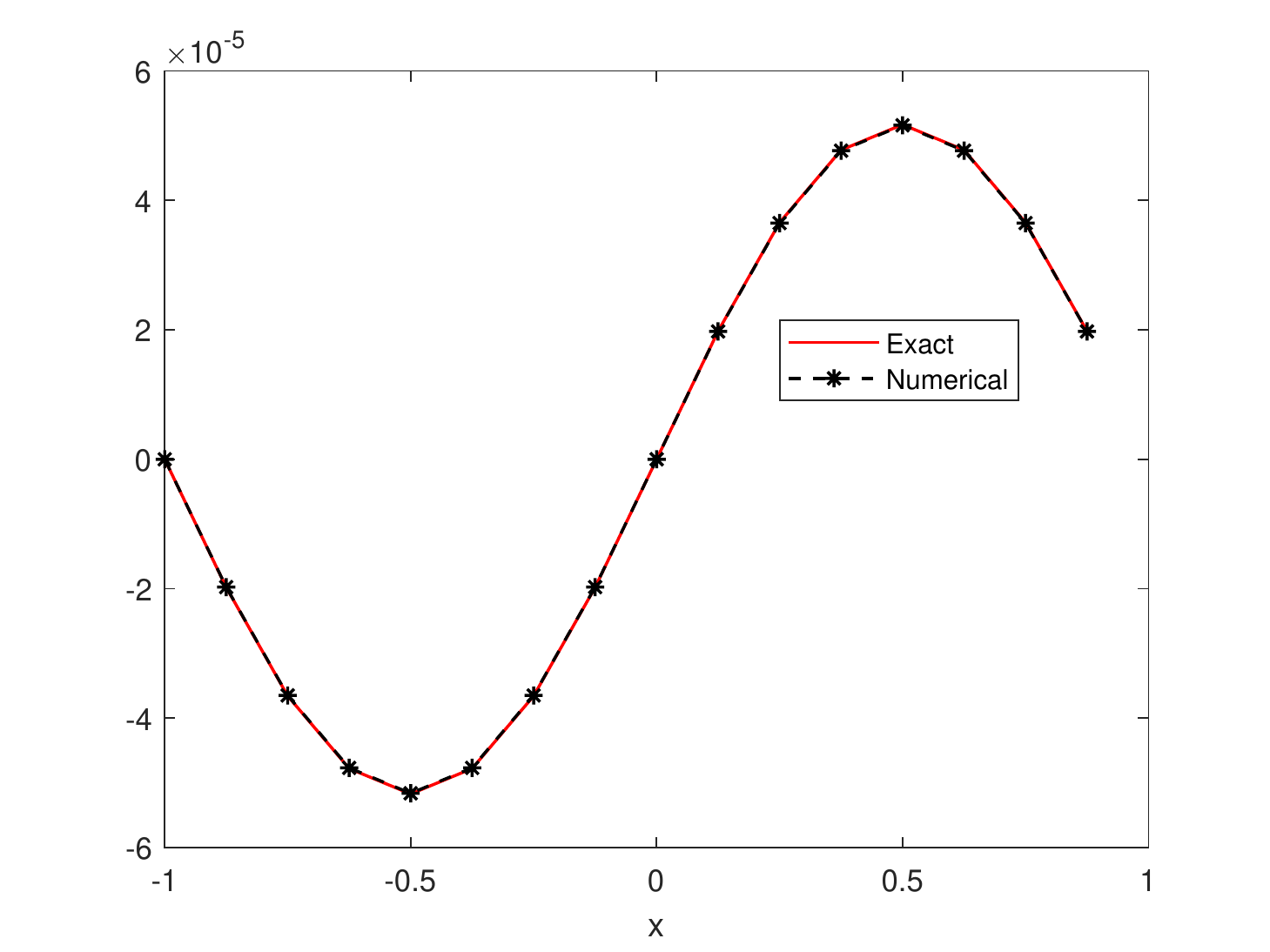}}
  \subfigure[waves in the frequency space]{\includegraphics[scale=0.35]{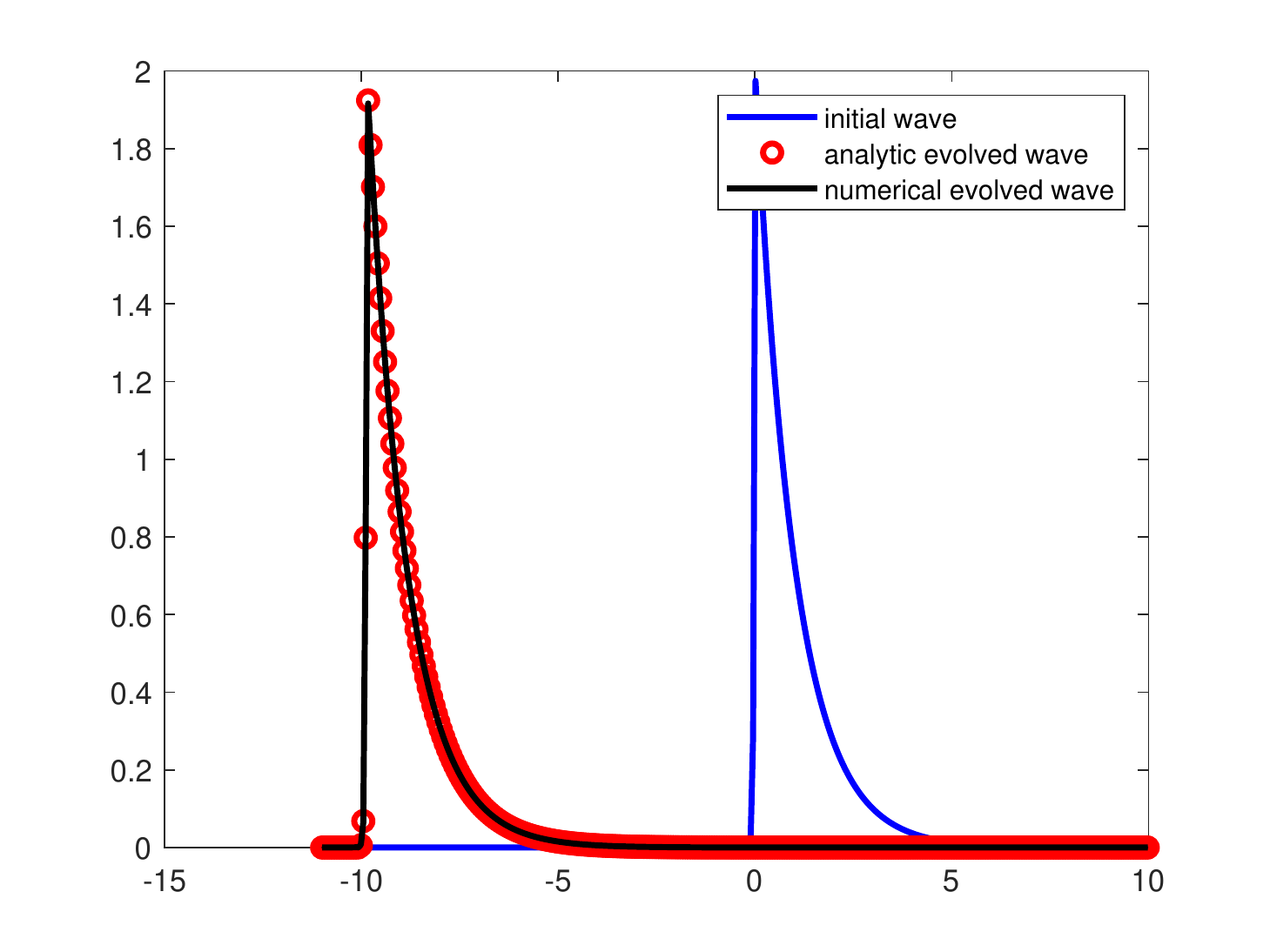}}
  \subfigure[waves in the original space]{\includegraphics[scale=0.35]{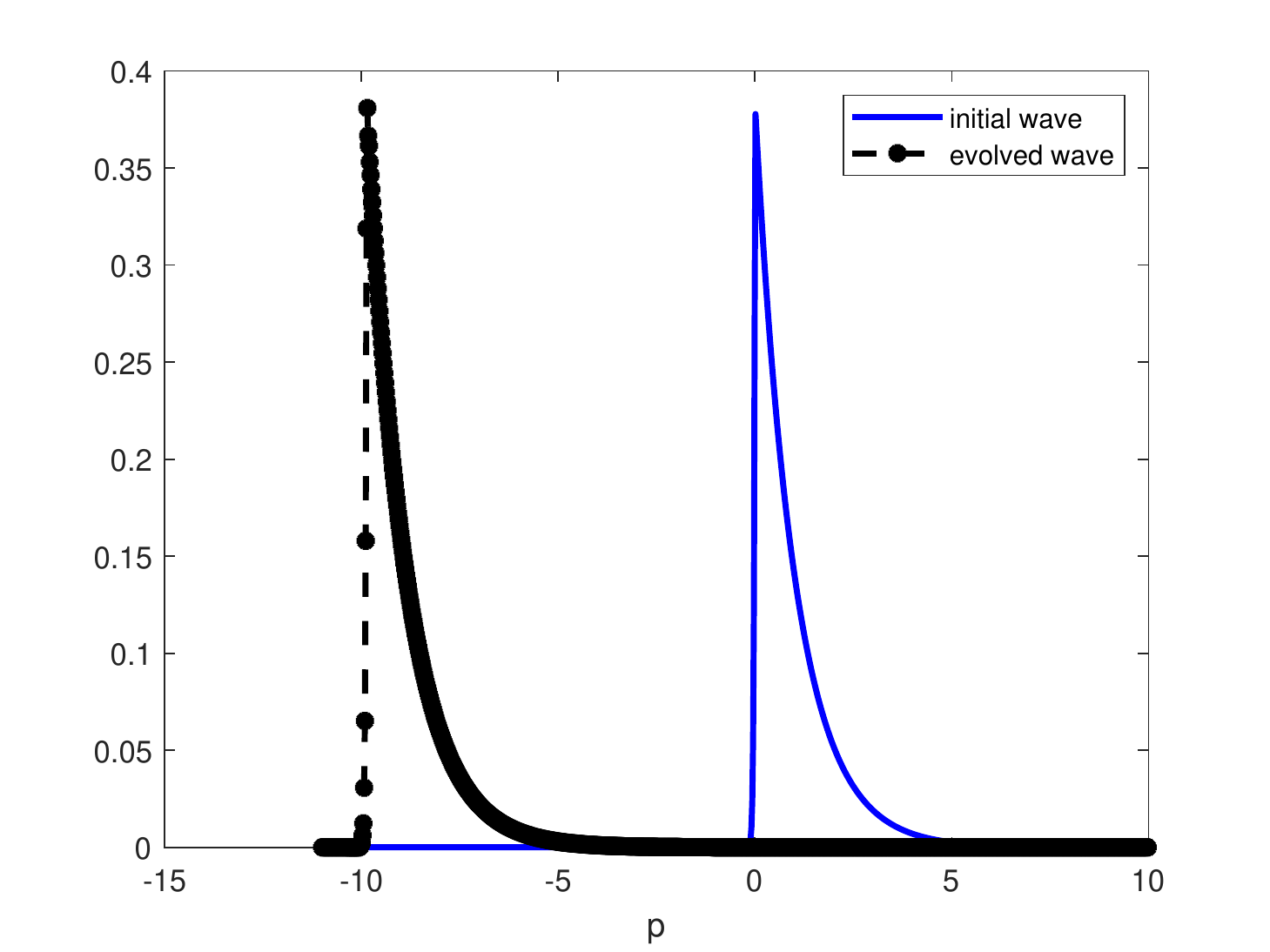}}
  \\
  \caption{Numerical results for given evolution time $t = T = 1$. (a) the numerical solution for the heat equation recovered by \eqref{integration}; (b) the modulus of the wave amplitude $|\hat{\bb{w}}_{l_*}(t,p_j)|$; (c) the wave amplitude $\bb{w}_{i_*}(t,p_j)$.}\label{fig:wwres}
\end{figure}

One can also use the finite difference discretisation to further validate the above arguments. For the spatial discretisation, we use the central difference to get
\begin{align*}
& \partial_t w_i(t,p) + \partial_p \frac{w_{i-1}(t,p)-2w_i(t,p) + w_{i+1}(t,p) }{\Delta x^2}  = 0 , \qquad i = 1,2,\cdots, M-1, \\
& w_0(t,p) = w_M(t,p), \qquad \frac{w_1(t,p) - w_{-1}(t,p)}{2\Delta x} = \frac{w_{M+1}(t,p) - w_{M-1}(t,p)}{2\Delta x},
\end{align*}
where we have introduced the ghost points $x_{-1}$ and $x_{M+1}$. To get a closed system, we assume the discretisation is valid on $x = x_0$ and $x=x_M$:
\begin{align*}
& \partial_t w_0(t,p) + \partial_p \frac{w_{-1}(t,p)-2w_0(t,p) + w_{1}(t,p) }{\Delta x^2}  = 0 , \qquad i = 0, \\
& \partial_t w_M(t,p) + \partial_p \frac{w_{M-1}(t,p)-2w_M(t,p) + w_{M+1}(t,p) }{\Delta x^2}  = 0, \qquad i = M.
\end{align*}
Summing the above two equations and eliminating the ghost values to get
\[\partial_t w_0(t,p) + \partial_p \frac{-2w_0(t,p) + w_{1}(t,p) + w_{M-1}(t,p) }{\Delta x^2}  = 0 , \qquad i = 0.\]
We then get the following system
\[\partial_t \bb{w}(t,p) + A \partial_p \bb{w}(t,p) = \bb{0}, \]
where
\[\bb{w}(t,p) = \begin{bmatrix}
w_0(t,p) \\ w_1(t,p) \\ \vdots \\ w_{M-2}(t,p) \\ w_{M-1}(t,p)
\end{bmatrix}, \qquad
A = \frac{1}{\Delta x^2} \begin{bmatrix}
-2 & 1        &         &            & 1 \\
1  & -2       &    1    &            &    \\
   & \ddots   & \ddots  &   \ddots   & \\
   &          & 1       &  -2        &  1 \\
1  &          &         &   1        & -2
\end{bmatrix}. \]
One can check that the eigenvalues of $A$ are $\lambda_k(A) = -\frac{4}{\Delta x^2} \sin^2 \frac{k \pi}{M}$ for $k = 0,1,\cdots,M-1$.
Let $\bb{w}_j^n$ be the approximation to $\bb{w} (t_n, p_j)$. For the periodic boundary condition, the value $\bb{w}_0 = \bb{w}_N$ along the boundaries are unknown. We introduce $\bb{w}_N$ into the vector of grid values $\bb{W}^n = [\bb{w}_1^n; \cdots; \bb{w}_N^n]$ with ``;" indicating the straightening of $\{\bb{w}_i\}_{i\ge 1}$ into a column vector. Since $\lambda_k \le 0$, we adopt the upwind discretisation on $p$ and obtain
\[\frac{\bb{w}_j^{n+1} - \bb{w}_j^n}{\Delta t} + A \frac{\bb{w}_{j+1}^n - \bb{w}_j^n}{\Delta p} = \bb{0},  \qquad j = 1,\cdots,N-1.\]
The above system is closed by assuming the discretisation holds at $p = p_0$. In view of the periodicity, the additional equation can be written as
\[\frac{\bb{w}_N^{n+1} - \bb{w}_N^n}{\Delta t} + A \frac{\bb{w}_1^n - \bb{w}_N^n}{\Delta p} = \bb{0},  \qquad j = N.\]
The final iterative system is
\[\bb{W}^{n+1} =  B \bb{W}^n, \qquad n = 0,1,\cdots, N_t - 1,\]
with
\[\bb{B} = \begin{bmatrix}
I+A_1 & -A_1        &         &            &   \\
      & I+A_1       &   -A_1  &            &    \\
      &             & \ddots  &   \ddots   & \\
      &             &         &  I+A_1     &  -A_1 \\
-A_1  &             &         &            & I+A_1
\end{bmatrix}, \qquad A_1 = \frac{\Delta t}{\Delta p}A. \]

The numerical results are similar to that of the spectral method under the same settings, as shown in Fig.~\ref{fig:fdmp}.

\begin{figure}[H]
  \centering
  \subfigure[use \eqref{point}]{\includegraphics[scale=0.5]{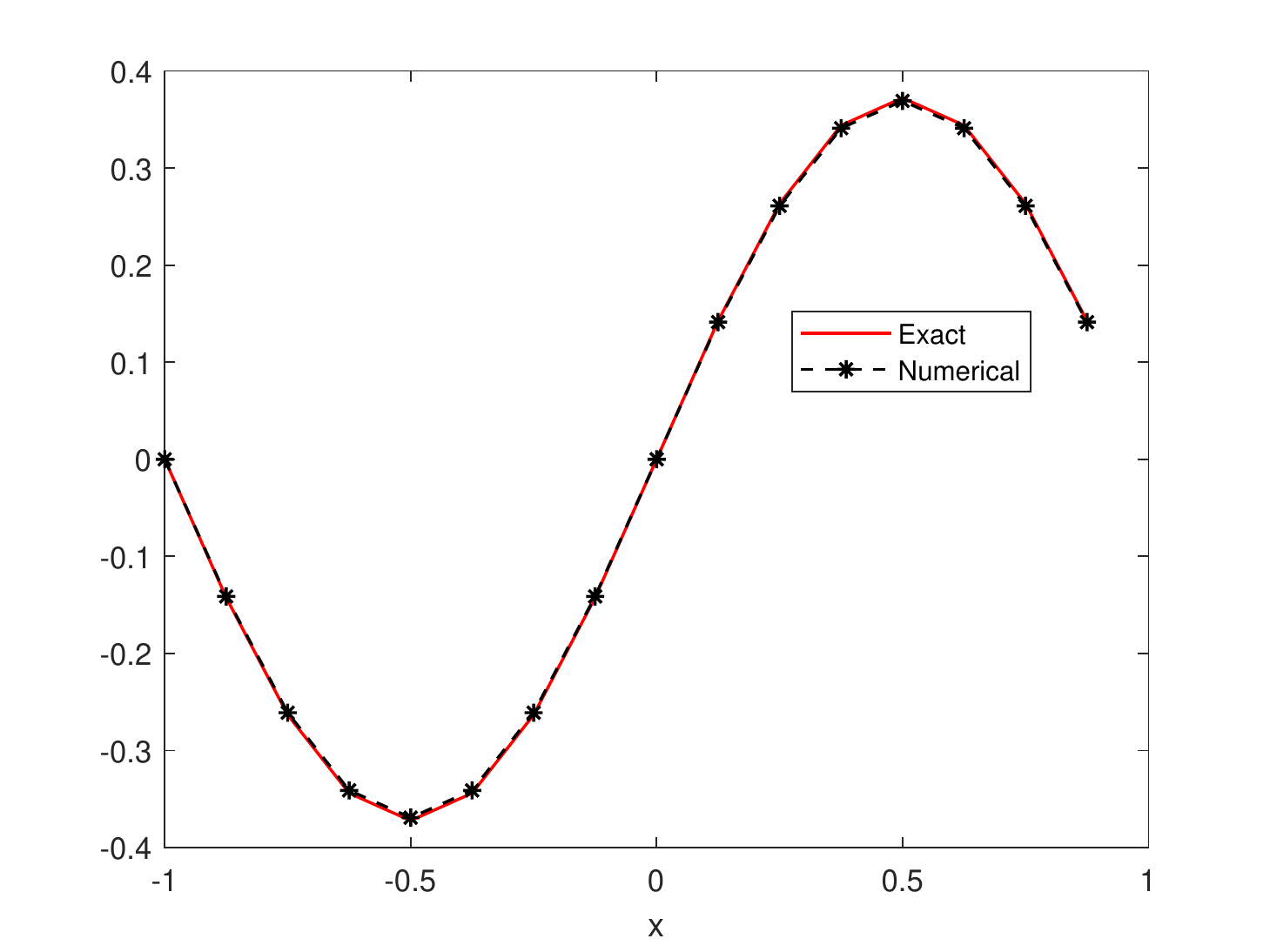}}
  \subfigure[use \eqref{integration}]{\includegraphics[scale=0.5]{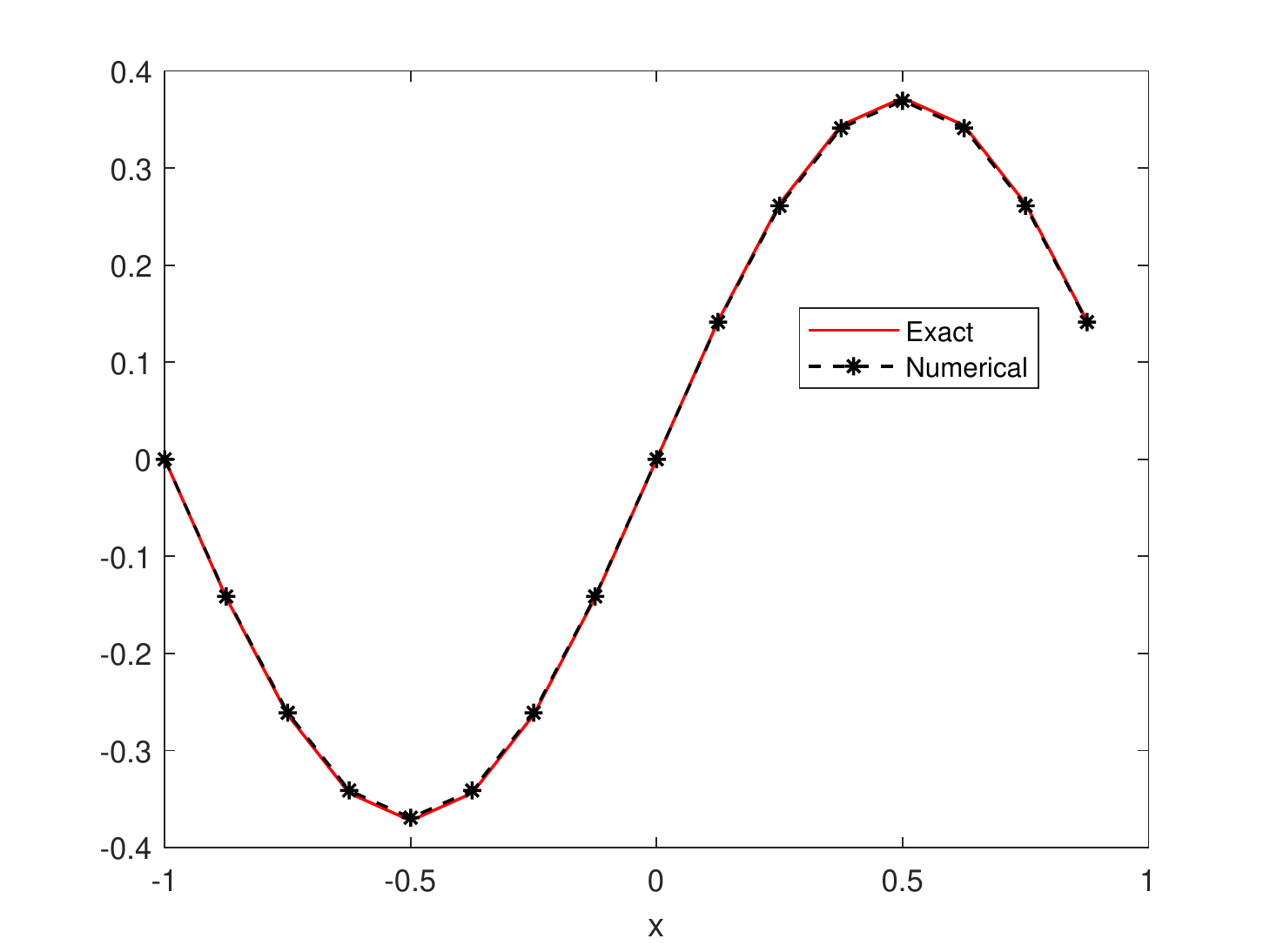}}\\
  \caption{Numerical and exact solutions of $\bb{u}(t = T_*)$ for the finite difference method.}\label{fig:fdmp}
\end{figure}

\subsubsection{The notations for the Fourier spectral discretisation}\label{subsect:notation}

For the Fourier spectral discretisation, we will use the notations given in \cite{JinLiuYu2022nonlinear}. For one-dimensional problems we choose a uniform spatial mesh size $\Delta x = (b-a)/M$  for $M=2N = 2^m$  with $m$ a positive integer, and  time step $\Delta t$, and we let the grid points be
\[x_j = a + j \Delta x, \quad j = 0,1,\cdots, M.\]
As an example we consider the periodic boundary conditions.
For $x\in [a,b]$, the 1-D basis functions for the Fourier spectral method are usually chosen as
\[\phi_l(x) = \e ^{\i \mu_l (x-a)} , \quad \mu_l = \frac{2\pi l}{b-a}, \quad  1 = -N,\cdots, N-1.\]
For convenience, we adjust the index as
\[\phi_l(x) = \e ^{\i \mu_l (x-a)} , \quad \mu_l = \frac{2\pi (l-N-1) }{b-a}, \quad 1 \le l \le M = 2N.\]
The approximation to $u(x)$ in the 1-D space is
\begin{equation}\label{Fexpand}
u(x) = \sum\limits_{l=1}^M c_l \phi_l(x), \quad x = x_j, ~~ j = 0,1,\cdots, M-1,
\end{equation}
which can be  written in vector form, $\bb{u} = \Phi \bb{c}$, where
\[\bb{u} = (u(x_j))_{M\times 1}, \quad \bb{c} = (c_l)_{M\times 1}, \quad
\Phi = (\phi_{jl})_{M\times M} = (\phi_l(x_j))_{M\times M}.\]

The $d$-dimensional grid points are given by ${x}_{\bb{j}} = (x_{j_1}, \cdots, x_{j_d})$, where $\bb{j} = (j_1,\cdots,j_d)$, and
\[x_{j_i} = a + j_i \Delta x, \quad j_i = 0,1,\cdots, M-1, \quad i = 1,\cdots,d.\]
The multi-dimensional basis functions are written as
$\phi_{\bb{l}}({x}) = \phi_{l_1}(x_1)\cdots \phi_{l_d}(x_d)$,
where $\bb{l} = (l_1,\cdots, l_d)$ and $1 \le l_i \le M$. The corresponding approximate solution is
$u(x) = \sum\nolimits_{\bb{l}} c_{\bb{l}} \phi_{\bb{l}}({x})$,
with the coefficients determined by the  values at the grid or collocation points $x_{\bb{j}} $. These collocation values will be arranged as a column vector:
\[\bb{u}(t) = \sum\limits_{\bb{j}}u(t,{x}_{\bb{j}}) \ket{j_1} \otimes \cdots \otimes \ket{j_d}.\]
That is, the $n_{\bb{j}}$-th entry of $\bb{u}$ is $u(t,{x}_{\bb{j}})$, with the global index given by
\[n_{\bb{j}}: = j_12^{d-1} + \cdots + j_d2^0, \qquad \bb{j} = (j_1,\cdots,j_d). \]
Similarly $c_{\bb{l}}$ is written in a column vector as $\bb{c} = \sum\nolimits_{\bb{l}} c_{\bb{l}} \ket{l_1} \otimes \cdots \otimes \ket{l_d}$.
For convenience, let $c_{\bb{l}} = c_{l_1}\cdots c_{l_d}$. Then,
\begin{equation}\label{Fexpandhigh}
u(t,{x}_{\bb{j}}) = \sum\limits_{\bb{l}} c_{l_1}\cdots c_{l_d} \phi_{l_1}(x_{j_1})\cdots\phi_{l_d}(x_{j_d}),
\end{equation}
and the transition between $\bb{u}$ and $\bb{c}$ is given by
\begin{equation}
\bb{u} = \bb{u}^{(1)} \otimes \cdots \otimes \bb{u}^{(d)}
 = ( \Phi \bb{c}^{(1)}) \otimes \cdots \otimes ( \Phi \bb{c}^{(d)})
 = \Phi^{\otimes ^d } \bb{c}, \label{usplit}
 \end{equation}
where
\[\Phi^{\otimes ^d } = \underbrace{\Phi \otimes \cdots \otimes \Phi}_{\text{$d$ matrices}},\qquad
\bb{c}^{(i)} = ( c_{l_i} )_{M\times 1}, \qquad \bb{u}^{(l)} = \Phi \bb{c}^{(l)},\]
\begin{equation}\label{cdef}
\bb{c} = \bb{c}^{(1)} \otimes \cdots \otimes \bb{c}^{(d)} = \sum\limits_{\bb{l}} c_{\bb{l}} \ket{l_1} \otimes \cdots \otimes \ket{l_d}.
\end{equation}

For later use, we next determine the transitions between the position operator $\hat{x}_j$ and the momentum operator $\hat{P}_j = -\i \frac{\partial}{\partial x_j}$ in discrete settings. Let $u(x)$ be a function in one dimension and $\bb{u} = [u(x_0),\cdots,u(x_{M-1})]^T$ be the mesh function with $M=2N$. The discrete position operator $\hat{x}^{\rm d}$ of $\hat{x}$ can be defined as
\[\hat{x}^{\rm d}: \bb{u} = \Big( u(x_i) \Big) \quad \to \quad
  \Big( x_i u(x_i) \Big) = D_x \bb{u}  \qquad \mbox{or} \qquad
  \hat{x}^{\rm d}\bb{u} = D_x \bb{u},\]
where $D_x = \text{diag} ( x_0, x_1, \cdots, x_{M-1} )$ is the matrix representation of the position operator in $x$-space. By the discrete Fourier expansion in \eqref{Fexpand}, the momentum operator can be discretised as
\begin{align*}
\hat{P}u(x)
& \approx \hat{P} \sum\limits_{l=1}^M c_l \phi_l(x) = \sum\limits_{l=1}^M c_l \hat{P} \phi_l(x)
  = \sum\limits_{l=1}^M c_l (-\i \partial_x \phi_l(x)) \\
& = \sum\limits_{l=1}^M c_l \mu_l \phi_l(x),  \quad \mu_l = 2\pi (l-N-1)
\end{align*}
for $x = x_j$, $j = 0,1,\cdots, M-1$, which is written in matrix form as
\[\hat{P}^{\rm d} \bb{u} =  \Phi D_\mu \Phi^{-1} \bb{u} =: P_\mu \bb{u}, \qquad
D_\mu = \text{diag} ( \mu_1, \cdots, \mu_M ),\]
where $\hat{P}^{\rm d}$ is the discrete momentum operator. The matrices $D_\mu$ and $P_\mu$ can be referred to as the matrix representation of the momentum operator in the momentum space and  the position space, respectively,  and are related by the discrete Fourier transform.
For $d$ dimensions, the discrete position operator $\hat{x}_l^{\rm d}$ is defined as
\[\hat{x}_l^{\rm d}: \bb{u} = \bb{u}^{(1)} \otimes \cdots \otimes \bb{u}^{(d)} \quad \to \quad
\bb{u}^{(1)} \otimes \cdots \otimes \tilde{\bb{u}}^{(l)} \otimes \cdots \otimes \bb{u}^{(d)}, \]
where
\[\tilde{\bb{u}}^{(l)}: = \Big( x_{j_{l_i}} u^{(l)} (x_{j_{l_i}}) \Big) = D_x \bb{u}^{(l)}.\]
Then,
\[\hat{x}_l^{\rm d} \bb{u} = (I^{\otimes^{l-1}} \otimes D_x \otimes I^{\otimes^{d-l}}) \bb{u} =: \bb{D}_l \bb{u}. \]
Using the expansion in \eqref{Fexpandhigh}, one easily finds that
\[\hat{P}_l^{\rm d} \bb{u} = (I^{\otimes^{l-1}} \otimes P_\mu \otimes I^{\otimes^{d-l}}) \bb{u} =: \bb{P}_l \bb{u}. \]
Note that $\bb{D}_l$ and $\bb{P}_l$ are Hermitian matrices, and
\begin{equation}\label{PlDl}
(\Phi^{\otimes^d})^{-1} \bb{P}_l \Phi^{\otimes^d} = I^{\otimes^{l-1}} \otimes D_\mu \otimes I^{\otimes^{d-l}}=:\bb{D}^\mu _l.
\end{equation}
For convenience, we denote by $F_x = \Phi^{\otimes^d}$ the Fourier transformation matrix in $d$ dimensions.
The above notations also apply to the variable $p$. One just needs to modify the subscript $x$ to $p$, for example, $F_p$ represents the Fourier transformation matrix for $p$. We shall use the same notations $D_\mu$ and $P_\mu$ for both $x$ and $p$ whenever no confusion will arise.

\begin{remark}\label{rem:FPhi}
Given a set of numbers $x_0, x_1, \cdots, x_{M-1}$, the discrete Fourier transform (DFT) and the inverse DFT are defined by
\[y_k = \frac{1}{\sqrt{M}}\sum\limits_{j=0}^{M-1} \e^{2\pi{\rm i} jk/M}x_j, \quad k = 0,\cdots, M-1\]
and
\[x_j = \frac{1}{\sqrt{M}}\sum\limits_{k=0}^{M-1} \e^{-2\pi{\rm i} jk/M}y_k, \quad j = 0,\cdots, M-1,\]
respectively. Denote the transformation matrix of DFT by $F$. It is easy to find the transformation matrix introduced above satisfies $\Phi = \sqrt{M} SF$, where $S = {\rm diag} \Big( [1, -1, \cdots, 1,-1]_{M\times 1} \Big)$ is a diagonal matrix.
\end{remark}

\subsubsection{The Schr\"odingerisation of the heat equation}

Eq.~\eqref{heat-Schro} clearly shows $w$ is governed by the free Schr\"odinger equation in the Fourier space for $p$. Below, in the discrete space, and for the heat equation with a source term, we also show the Schr\"odingerisation of the heat equation.
Consider the linear heat equation with a source term:
\begin{equation} \label{heateq}
\begin{cases}
\partial_tu - \Delta u = V(x)u, \\
u(0,x) = u_0(x),
\end{cases}
\end{equation}
where $V(x)$ is a scalar function. We now consider the spectral discretisation of \eqref{heatreformulation} with respect to the variable $p$. According to the above discussion, the term $\partial_p w$ in \eqref{heatreformulation} can be discretised as
\[\partial_p w =  \i \cdot (-\i \partial_p w) \longrightarrow \i P_\mu \bb{w}(t,x),\]
where $\bb{w}(t,x) = [w(t,x,p_0), w(t,x,p_1), \cdots, w(t,x,p_{M-1})]^T$. Thus we have the semi-discrete system of partial differential equations
\begin{equation*}
\begin{cases}
\partial_t \bb{w}(t,x) + \i ( \Delta_x + V(x) ) P_\mu \bb{w}(t,x) = \bb{0},  \\
\bb{w}(0,x) = \bb{w}_0(x):=(w_0(x,p_j) )_{M \times 1}.
\end{cases}
\end{equation*}
In the momentum basis, one has
\begin{equation}\label{Schrodingeriseheat}
\begin{cases}
\partial_t \bb{c}(t,x) + \i ( \Delta_x + V(x)) D_\mu \bb{c}(t,x) = \bb{0}, \\
\bb{c}(0,x) = \bb{c}_0(x):= \Phi^{-1}\bb{w}_0(x),
\end{cases}
\end{equation}
where $\bb{c} = F_p^{-1} \bb{w}$. The above PDEs can be viewed as the Schr\"odingerised system since we have introduced the imaginary unit $\i = \sqrt{-1}$ and $D_\mu= \text{diag} ( \mu_1, \cdots, \mu_M )$ is a real diagonal matrix.

In conclusion, we have converted the heat equation to a decoupled system of Schr\"odinge equation  by using the idea in Subsect.~\ref{subsect:idea} and the discrete Fourier transform on $p$.

\begin{remark}
 The Schr\"odinger equation is $\i \partial_t\Psi =-\frac{\hbar}{2m} (\Delta-\frac{2m}{\hbar^2} \tilde{V})\Psi$. To match the form above, we can interpret $\mu_j \longleftrightarrow -\hbar/(2m)$ and $V_j=-\tilde{V}/(\hbar \mu_j)$, which is the potential corresponding to the $j^{\text{th}}$ mode. However, this requires us to have the right sign $\mu_j<0$, otherwise we need to interpret ``negative mass''.  This issue can be easily resolved by introducing the new variable $\hat{\bb{c}}_j(\hat{t},\cdot) = \bb{c}_j(t,\cdot)$, where $\hat{t} = -t$ if $\mu_j>0$.
\end{remark}

\subsubsection{The quantum simulation of the heat equation} \label{subsect:heat}

One can directly simulate the (decoupled) system of  Schr\"odingerised  equations \eqref{Schrodingeriseheat}. For clarity, we will consider $(x,p)$ as a new variable and repeat the construction in the previous subsection. Let
\begin{equation}\label{wapprox}
w(t,x,p) = \sum\limits_{\bb{l}}c_{\bb{l}}(t) \phi_{\bb{l}}(x,p), \quad
\bb{l} = (l_1, \cdots, l_d, l_p).
\end{equation}
The collocation points are denoted by $(x_{\bb{j}}, p_{j_p})$ with $\bb{j} = (j_1,\cdots,j_d)$. As in \eqref{cdef}, we define $\bb{c} = \bb{c}_x \otimes \bb{c}_p$, where $\bb{c}_x = \bb{c}^{(1)} \otimes \cdots \otimes \bb{c}^{(d)}$.
We also introduce the notation $\bb{w} = \bb{w}_x \otimes \bb{w}_p$, where
$\bb{w}_x = \bb{w}^{(1)} \otimes \cdots \otimes \bb{w}^{(d)}$, and $\bb{w}^{(l)} = \Phi \bb{c}^{(l)}$ can be viewed as the approximate solution of $w$ in $x_l$ direction.
Following the discussion in Subsect.~\ref{subsect:notation}, one has
\begin{align*}
\Delta_x w_p
& =  \sum\limits_{l=1}^d \partial_{x_l}^2 \partial_p w
  = -\i \sum\limits_{l=1}^d (-\i \partial_{x_l})^2  (-\i \partial_p) w
 =  -\i \sum\limits_{l=1}^d \hat{P}_l^2  \hat{P}_p w
\longrightarrow  -\i \sum\limits_{l=1}^d (\hat{P}^{\text{d}}_l)^2 \hat{P}^{\text{d}}_p \bb{w} \\
& =  -\i \sum\limits_{l=1}^d (\bb{P}_l^2 \otimes I)(I^{\otimes^d} \otimes P_\mu) (\bb{w}_x \otimes \bb{w}_p)
 = -\i \sum\limits_{l=1}^d (\bb{P}_l^2 \otimes P_\mu ) \bb{w},
\end{align*}
where $\bb{P}_l = I^{\otimes^{l-1}} \otimes P_\mu \otimes I^{\otimes^{d-l}}$, and
\begin{align*}
 V(x) w_p
&  = \i V(x)  (-\i \partial_p) w
 \longrightarrow \i V(\hat{x}^{\text{d}})\hat{P}^{\text{d}}_p \bb{w} \\
& = \i (\bb{V} \otimes I)  (I^{\otimes^d} \otimes P_\mu) (\bb{w}_x \otimes \bb{w}_p)
 = \i (\bb{V} \otimes  P_\mu) \bb{w},
\end{align*}
where $\bb{V}$ is a diagonal matrix with
\[\bb{V}_{n_{\bb{j}}, n_{\bb{j}} } = V({x}_{\bb{j}}), \quad n_{\bb{j}} = j_12^{d-1} + \cdots +  j_d2^0. \]
The resulting ODEs are
\begin{equation}\label{heatw}
\frac{\d}{\d t} \bb{w}(t) - \i \Big( (\bb{P}_1^2 + \cdots + \bb{P}_d^2 - \bb{V} )\otimes P_\mu \Big)   \bb{w}(t) = \bb{0},
\end{equation}
which is an expected Hamiltonian system.

The system \eqref{heatw} can be solved by using, for example, the first-order time or Trotter splitting, as done for the standard Schr\"odinger equation in \cite{JinLiuYu2022nonlinear}. To do so, we first diagonalise the matrix with respect to $p$ for convenience as in \eqref{Schrodingeriseheat}. Introducing a new variable $\widetilde{\bb{w}} = (I^{\otimes^d} \otimes F_p^{-1}) \bb{w}$, one gets
\begin{equation}\label{heatHamiltonianDiag}
\frac{\d}{\d t} \widetilde{\bb{w}}(t) = \i \bb{H}  \widetilde{\bb{w}}(t), \qquad \bb{H} = ( \bb{P}_1^2 + \cdots + \bb{P}_d^2 - \bb{V} ) \otimes D_\mu.
\end{equation}
From time $t=t_n$ to time $t = t_{n+1}$, the system can be solved in two steps:
\begin{itemize}
  \item One first solves
  \[\begin{cases}
  \frac{\d}{\d t} \widetilde{\bb{w}}(t) = \i ( ( \bb{P}_1^2 + \cdots + \bb{P}_d^2 ) \otimes D_\mu)  \widetilde{\bb{w}}(t), \qquad t_n < t < t_{n+1}, \\
  \widetilde{\bb{w}}(t_n) = \widetilde{\bb{w}}^n
  \end{cases}\]
for one time step, where $\widetilde{\bb{w}}^n$ is the numerical solution at $t=t_n$. Noting the relation \eqref{PlDl}, by letting $\widetilde{\bb{c}}(t) = ( F_x^{-1} \otimes I) \widetilde{\bb{w}}(t)$, we instead solve
\[\begin{cases}
\frac{\d}{\d t} \widetilde{\bb{c}}(t) = \i \bb{H}_D  \widetilde{\bb{c}}(t),\qquad t_n < t < t_{n+1}, \\
\widetilde{\bb{c}}(t_n) = \widetilde{\bb{c}}^n = (F_x^{-1} \otimes I) \widetilde{\bb{c}}^n,
\end{cases}\]
  where
  \[\bb{H}_D = ( (\bb{D}_1^\mu)^2 + \cdots + (\bb{D}_d^\mu)^2 ) \otimes D_\mu\]
  is a diagonal matrix. The numerical solution will be denoted by $\tilde{\bb{c}}^*$.
  \item Let $\widetilde{\bb{w}}^* = ( F_x \otimes I) \widetilde{\bb{c}}^*$. The second step is to solve
  \[\frac{\d}{\d t} \widetilde{\bb{w}}(t) =  -\i (\bb{V}  \otimes D_\mu)  \widetilde{\bb{w}}(t) =: -\i \bb{H}_V \widetilde{\bb{w}}(t) \]
 for one time step, with $\widetilde{\bb{w}}^*$ being the initial data, where $\bb{H}_V $ is a diagonal matrix. This gives the updated numerical solution $\widetilde{\bb{w}}^{n+1}$.
\end{itemize}

 We denote $m$ to be the number of qubits per dimension. The total number of qubits is then given by $m_H = m_d + m_p$ for $d$-dimensional problems, where $m_d = dm$ and $m_p$ are the number of qubits on $x$ and $p$ registers, respectively. Let $\Delta x$, $\Delta p$ and $\Delta t$ be the step sizes for $x$, $p$ and $t$, respectively, where we assume the same spatial step along each dimension. Then $m \sim \log (1/\Delta x)$ and $m_p \sim \log (1/\Delta p)$. 

For convenience we write the time complexity in terms of the steps sizes and the number of qubits throughout the paper. For the given error bound $\varepsilon$, the $\varepsilon$-dependence of these quantities is determined by the particular scheme one wishes to use. For instance, the mesh strategy of the heat equation can be given by
\[\Delta x \sim (\varepsilon/d)^{1/\ell}, \qquad \Delta t \sim  \varepsilon, \qquad \Delta p \sim  \varepsilon.\]
Note that the initial condition--due to lack of regularity-- implies first-order accuracy on $p$. 

\begin{theorem} \label{thm:heat}
The solution to the heat equation can be simulated with gate complexity given by
\[N_{\text{Gates}} = T/\Delta t \cdot \mathcal{O}(d m \log m + m_p \log m_p) 
%= \frac{T }{\varepsilon} \cdot \mathcal{O}( d  \log \frac{d^{1/\ell}}{\varepsilon^{1/\ell}}  +  \log \frac{1}{\varepsilon})
,\]
where $T$ is the evolution time.
\end{theorem}
\begin{proof}
Given the initial state of $\widetilde{\bb{w}}^0$, applying the inverse QFT to the $x$-register, one gets $\widetilde{\bb{c}}^0$.
At each time step, one needs to consider the following procedure
\[\widetilde{\bb{c}}^n
\xrightarrow {\e^{\i \bb{H}_D \Delta t}} \widetilde{\bb{c}}^*
\xrightarrow {F_x\otimes I } \widetilde{\bb{w}}^*
\xrightarrow {\e^{-\i \bb{H}_V \Delta t}} \widetilde{\bb{w}}^{n+1}
 \xrightarrow {F_x^{-1}\otimes I } \widetilde{\bb{c}}^{n+1},\]
where $F_x = \Phi^{\otimes^d}$.
It is known that the quantum Fourier transforms in one dimension can be implemented using $\mathcal{O}(m \log m)$ gates. The diagonal unitary operators $\e^{-\i \bb{H}_V \Delta t}$ and $\e^{\i \bb{H}_D \Delta t}$ can be implemented using $\mathcal{O}(m_H)$ gates \cite{Kassal2008Diagonal,Jin2022quantumSchrodinger}. Therefore, the gate complexity required to iterate to the $n$-th step is
\begin{align*}
N_{\text{Gates}}
 = n  \mathcal{O}(d m \log m + m_p \log m_p + m_H  )
 = T/\Delta t \cdot \mathcal{O}(d m \log m + m_p \log m_p),
\end{align*}
where $m_p \log m_p$ is resulted from the QFT on the $p$ register, which is only performed twice.
%
%Note that the initial condition--due to lack of regularity-- implies the first order accuracy on $p$. The error of the spectral discretisation in $L^2$ norm is then given by
%\[\text{Err} \le  C(  \Delta t + d \Delta x^\ell + \Delta p),\]
%which leads to the following mesh strategy:
%\[\Delta x \sim (\varepsilon/d)^{1/\ell}, \qquad \Delta t \sim  \varepsilon, \qquad \Delta p \sim  \varepsilon\]
%by forcing all error terms to be of order $\varepsilon$. This completes the proof.
\end{proof}

\begin{remark}\label{rem:algI}
For the rest of the paper, we refer to the algorithm mentioned in \cite{Kassal2008Diagonal,Jin2022quantumSchrodinger} to implement $\e^{\i H t}$ as Algorithm I,  where $H$ is a diagonal matrix.
\end{remark}

\subsection{Schr\"odingerization of  the linear convection equation}\label{subsec:convection}

In this section we provide a way to Schr\"odingerise the linear convection equation
\begin{equation}\label{convection}
\partial_t u+ \partial_{x_1}u + \partial_{x_2}u + \cdots + \partial_{x_d} u= 0, \quad \bb{x} = (x_1,x_2,\cdots,x_d) \in (-1,1)^d
\end{equation}
in $d$ dimensions with periodic boundary conditions, where $u = u(t,x_1,x_2,\cdots, x_d)$.

\subsubsection{The reformulation}

Let
\[w = \sin(p) u,
\]
where $p\in [-\pi,\pi]$, which is obviously periodic with respect to $p$. Then $w$ satisfies
\[\partial_t w - \partial_{x_1,pp} u - \partial_{x_2,pp} u - \cdots - \partial_{x_d,pp} u = 0.\]
Considering the Fourier spectral discretisation on $x$, one easily gets
\[\partial_t \bb{w}(t,p) - \i \sum\limits_{l=1}^d \bb{P}_l \partial_{pp} \bb{w}(t,p) = \bb{0}, \]
where $\bb{w}(t,p) = \sum_{\bb{j}} w(t,x_{\bb{j}},p) \ket{\bb{j}}$. Let $\bb{c}_x(t,p) = F_x^{-1} \bb{w}(t,p)$. We then get the system of (decoupled) free Schr\"odinger-type equations
\begin{equation}\label{cxhyper}
\partial_t \bb{c}_x(t,p) - \i \sum\limits_{l=1}^d \bb{D}_l^\mu \partial_{pp} \bb{c}_x(t,p) = \bb{0}
\end{equation}
in the momentum space, where $\bb{D}_l^\mu$ is a diagonal matrix defined by \eqref{PlDl}.

\subsubsection{The quantum simulation}

By further applying the Fourier transform on $p$,  Eq.~\eqref{cxhyper} becomes
\[\partial_t \tilde{\bb{c}}_x(t) + \i \sum\limits_{l=1}^d (\bb{D}_l^\mu \otimes P_\mu^2) \tilde{\bb{c}}_x(t) = \bb{0}, \]
where
\[\tilde{\bb{c}}_x (t) = [\tilde{\bb{c}}_{x,0} (t); \cdots; \tilde{\bb{c}}_{x,M-1} (t)], \qquad \tilde{\bb{c}}_{x,i} (t) = \sum_k \bb{c}_{x,i}(t,p_k) \ket{k} .\]
 Let $\bb{c}(t) = (I^{\otimes^d} \otimes F_p^{-1})\tilde{\bb{c}}_x$. Then,
\begin{equation}\label{chyper}
\partial_t \bb{c}(t) + \i \sum\limits_{l=1}^d (\bb{D}_l^\mu \otimes D_\mu^2) \bb{c}(t) = \bb{0},
\end{equation}
where $\bb{c}$ is exactly the momentum variables of $\bb{w}$, i.e., $\bb{c} = (F_x^{-1} \otimes F_p^{-1}) \bb{w}$.

\begin{theorem} \label{thm:hyper1}
The solution to the convection equation can be simulated with gate complexity
\[N_{\text{Gates}} = \mathcal{O}((d+1) m \log m),\]
where we have assumed the same number of qubits along every direction.
\end{theorem}
\begin{proof}
Given the initial state of $\bb{w}^0$, the implementation involves an application of an inverse quantum Fourier Transform (QFT), followed by a multiplication of a diagonal unitary operator $\bb{H} = \sum\limits_{l=1}^d (\bb{D}_l^\mu \otimes D_\mu^2)$ and a QFT. Hence the gate complexity (see Algorithm I in Remark \ref{rem:algI}) is
\begin{align*}
N_{\text{Gates}}  =  \mathcal{O}( m_{d+1}  +  2 (d+1) m \log m ) = \mathcal{O}((d+1)  m \log m ),
\end{align*}
which completes the proof.
\end{proof}

\begin{remark}
As was done for the Liouville equation in \cite{JinLiu2022nonlinear},
one can apply the Fourier spectral discretisation to the original equation and gets
\[
\partial_t \bb{c}_x(t) + \i \sum\limits_{l=1}^d \bb{D}_l^\mu \bb{c}_x(t) = \bb{0},
\]
 where $\bb{c}_x = F_x^{-1} \bb{u}$ and $\bb{u}(t) = \sum_{\bb{j}} u(t,x_{\bb{j}}) \ket{j}$. The number of quantum gates required for its simulation is $\mathcal{O}(d m \log m)$, which is comparable to that for the Schr\"odingerised system.
 On the other hand,
 one can get $\bb{u}$ by projecting $\bb{w}$ onto the $x$-register for the Schr\"odingersiation approach since $w(t,x,p) = \sin(p) u(t,x)$ is separated in $x$ and $p$, which means the costs of the computation of the observables are also comparable. Here we want to show that --for conceptual interest--the first order hyperbolic equation can be transformed into a system of Schr\"odinger-type equations.
\end{remark}

\section{Quantum simulation of general linear system of ODEs} \label{sec:general}

Not every linear PDE can be transformed into the exact form of the Schr\"odinger equation, but we will see that they can certainly be transformed into Hamiltonian systems in the discrete setting. To do so, one just needs to show that a linear system of ordinary differential equations can be converted into a Hamiltonian system, since a PDE, after spatial discretisation, becomes a system of ODEs.

\subsection{Schr\"odingerisation of linear ODEs} \label{subsect:generalisation}

Suppose one needs to solve the following ODEs
\begin{equation}\label{ODElinear}
 \begin{cases}
 \frac{\d \bb{u}(t)}{\d t} = A \bb{u}(t) + \bb{b}(t), \\
 \bb{u}(0) = \bb{u}_0,
 \end{cases}
\end{equation}
where matrix  $A$ is independent of time, and  $A^{\dagger} \neq A$ in general. We remark that it suffices to assume $\bb{b}(t) = \bb{0}$. Otherwise one can instead consider the augmented system:
\[ \begin{cases}
 \frac{\d \bb{u}(t)}{\d t} = A \bb{u}(t) + \bb{b}(t) v,  \qquad \bb{u}(0) = \bb{u}_0,\\
 v_t = 0, \qquad v(0) = 1,
 \end{cases}\]
where the second equation gives $v(t) \equiv 1$, which leads to the original ODE system.  The above ODEs can be written in the following compact form
\[ \begin{cases}
 \frac{\d \tilde{\bb{u}}(t)}{\d t} = \tilde{A} \tilde{\bb{u}}(t)  \\
 \tilde{\bb{u}}(0) = \tilde{\bb{u}}_0
 \end{cases},\qquad \tilde{\bb{u}} = \begin{bmatrix} \bb{u} \\ v \end{bmatrix},  \qquad\tilde{A} = \begin{bmatrix}
A & \bb{b}(t) \\
\bb{0}^T & 0
\end{bmatrix}, \qquad
\tilde{\bb{u}}_0 = \begin{bmatrix} \bb{u}_0 \\ 1 \end{bmatrix},\]
where the zero vector $\bb{0}$ has the same size as $\bb{b}$. For this reason, without loss of generality,  we assume $\bb{b} = \bb{0}$ in the following.

\begin{remark} \label{rem:augment}
 Quantum simulations for time-dependent or independent boundary value problems are quite difficult because the ODE system resulting from spatial discretisations is not necessarily a Hamiltonian system. Spatial discretization of the boundary condition, for example the Dirichlet boundary condition, could also give rise to  $\bb{b} \not= \bb{0}$. Our  Schr\"odingerisation approach combined with the above augmentation technique will resolve this problem in a generic and an efficient way. This will be further investigated in our forthcoming work.
\end{remark}

One can always decompose $A$ into a Hermitian term and an anti-Hermitian term:
\[A = H_1 + \i H_2,\]
where
\[
H_1 = \frac{A+A^{\dagger}}{2} = H_1^{\dagger}, \qquad H_2 = \frac{A-A^{\dagger}}{2 \i} = H_2^{\dagger}.
\]
Apparently, $H_1 H_2 = H_2 H_1$ holds if and only if $A^\dag A = A A^\dag$. In view of the stability, it is natural to assume that $H_1$ is negative semi-definite (note that $\bb{x}^\dag H_1 \bb{x} = \bb{x}^\dag A \bb{x}$).  In addition, it is simple to see that
\[2s(H_1) \sim s(A) \sim 2s(H_2) \qquad \mbox{and} \qquad \|H_1\|_{\max}, \|H_2\|_{\max} \le \|A\|_{\max},\]
where $s(A)$ is the sparsity of $A$ and $\|A\|_{\max}$ denotes the largest entry of $A$ in absolute value.

Using the warped phase transformation $\bb{v}(t,p) = \e^{-p} \bb{u}(t)$ for $p>0$ and symmetrically extend the initial data to $p<0$ as in Subsect.~\ref{subsect:idea}, the ODEs  are then transferred to a system of linear convection equations:
\begin{equation}\label{u2v}
\begin{cases}
 \frac{\d}{\d t} \bb{v}(t,p) = A \bb{v}(t,p) = - H_1 \partial_p \bb{v} + \i H_2 \bb{v}, \\
 \bb{v}(0,p) = \e^{-|p|} \bb{u}_0.
 \end{cases}
\end{equation}
The solution $\bb{u}(t)$ can be restored by
\[\bb{u}(t) = \int_0^\infty \bb{v}(t,p) \d p.\]
Apply the discrete Fourier transformation on $p$ to get
\begin{equation}\label{heatww}
\frac{\d}{\d t} \bb{w}(t) = -\i ( H_1 \otimes P_\mu ) \bb{w} + \i (H_2 \otimes I) \bb{w},
\end{equation}
which is a Hamiltonian system as expected, where $\bb{w}$ collects all the grid values of $\bb{v}(t,p)$ and is defined by
\[\bb{w} = [\bb{w}_1; \bb{w}_2; \cdots; \bb{w}_n], \qquad \bb{w}_i = \sum_k \bb{v}_i (t,p_k) \ket{k},\]
with ";" indicating the straightening of $\{\bb{w}_i\}_{i\ge 1}$ into a column vector. By the change of variables $\tilde{\bb{w}} = (I_u \otimes F_p^{-1})\bb{w}$, one has
\begin{equation}\label{generalSchr}
\frac{\d}{\d t} \tilde{\bb{w}}(t) = -\i ( H_1 \otimes D_\mu ) \tilde{\bb{w}} + \i (H_2 \otimes I) \tilde{\bb{w}}.
\end{equation}

In the sequel, we assume $A$ is independent of time and apply the Hamiltonian simulation algorithm in \cite{Berry-Childs-Kothari-2015} as described in the following lemma. For time-dependent Hamiltonians, we refer the reader to \cite{BerryChilds2020TimeHamiltonian,An2022timedependent,AnLin2021TimeHamiltonian} for references.  As in \cite{Berry-Childs-Kothari-2015}, we are concerned with the sparse access to the matrix, with the definition given below.

From now on we define $s$ as the sparsity of $A$, and $\|A\|_{\max}$ as the largest entry of $A$ in absolute value,

\begin{definition}
Let $A$ be a Hermitian matrix with the $(i,j)^{\text{th}}$ entry denoted by $A_{ij}$. Sparse access to $A$ is referred to as a 4-tuple $(s, \| A \|_{\text{max}}, O_A, O_F)$. Here $O_A$ is a unitary black box which can access the matrix elements $A_{ij}$ such that
\[O_A |j\rangle|k\rangle|z\rangle = |j\rangle|k\rangle|z\oplus A_{jk}\rangle\]
for any $j,k \in \{1,2,\cdots,N\}=:[N]$, where the third register holds a bit string representing of $A_{jk}$;
$O_F$ is a unitary black box which allows to perform the map
\[O_F |j\rangle|l\rangle = |j\rangle|F(j,l)\rangle\]
for any $j\in [N]$ and $l \in [s]$, where the function $F$ outputs the column index of the $l^{\text{th}}$ non-zero elements in row $j$.
\end{definition}

The quantum algorithm for general sparse Hamiltonian simulation with nearly optimal dependence on all parameters can be found in \cite{Berry-Childs-Kothari-2015}. The run time is measured in terms of the query complexity or the number of queries made to the oracles $O_A$ and $O_F$.

\begin{lemma}[Algorithm II, Theorems 1-2 in \cite{Berry-Childs-Kothari-2015}]\label{lem:complexityHamiltonian}
An $s$-sparse Hamiltonian $H$ acting on $m_H$ qubits can be simulated within error $\varepsilon$ with
\[\mathcal{O}\Big(  \tau \frac{\log (\tau/\varepsilon) }{\log\log (\tau/\varepsilon)} \Big)\]
queries and
\[\mathcal{O}\Big(  \tau ( m_H + \log^{2.5}(\tau/\varepsilon) )\frac{\log (\tau/\varepsilon) }{\log\log (\tau/\varepsilon)} \Big)
= \mathcal{O}( \tau m_H \cdot \text{polylog})\]
additional 2-qubits gates, where $\tau = s \|H\|_{\max} t$ and $t$ is the evolution time, and
\[ \text{polylog}  \equiv  \log^{2.5}(\tau/\varepsilon) \frac{\log (\tau/\varepsilon) }{\log\log (\tau/\varepsilon)}.\]
 This result is near-optimal.
\end{lemma}

For convenience, we assume $A$ arises from some discretisation in $d$-dimensions, which implies $A$ is of order $N_A \sim 2^{m_d}$, where $m_d = d m$ and $m$ is the number of qubits along each direction.
\begin{theorem}\label{thm:Schrodingerisation}
The Schr\"odingerisation method has gate complexity
\begin{align*}
N_{\text{Gates,Schr}}
 = (m_d + m_p)  \widetilde{\mathcal{O}}( s(A)\|A\|_{\max}/\Delta p ) + \mathcal{O}( m_p \log m_p),
\end{align*}
where $N_A$ is the order of $A$.
In particular, if $H_1$ can be diagonalised in the momentum basis and $H_2$ is a diagonal matrix, then the Schr\"odingerisation method can be simulated with
\[N_{\text{Gates,Schr}} = T/\Delta t \cdot  \mathcal{O}( d m \log m + m_p \log m_p).\]
\end{theorem}
\begin{proof}
1) The Hamiltonian in \eqref{generalSchr} can be simulated with
\begin{align*}
&  (m_d + m_p)  \widetilde{\mathcal{O}}( s( H_1\otimes D_\mu - H_2 \otimes I) \|H_1\otimes D_\mu - H_2 \otimes I\|_{\max} ) \\
& \le  (m_d + m_p)  \widetilde{\mathcal{O}}(s(H_1\otimes D_\mu)+s(H_2))(\|H_1 \otimes D_{\mu}\|_{\max}+\|H_2\|_{\max}) )\\
& \le (m_d + m_p)  \widetilde{\mathcal{O}}\Big( (s(H_1)+s(H_2))(\|H_1\|_{\max}/\Delta p+\|H_2\|_{\max}) \Big) \\
& \le (m_d + m_p)  \widetilde{\mathcal{O}}( s(A)\|A\|_{\max}/\Delta p ),
\end{align*}
where $\|D_\mu\|_{\max} \lesssim 1/\Delta p$ is used. 

2) For the special case, one can solve \eqref{generalSchr} by the first-order time splitting scheme.
The gate complexity can be obtained as in the proof of Theorem \ref{thm:heat}.
\end{proof}

The difference between this method and our approach in Section \ref{sect:Schrodingerisation} is that here we first discretise in space (hence $A$ is the difference matrix) and then use the warped phase transformation later, but in Section \ref{sect:Schrodingerisation} we use the warped phase transformation first and then discretise in space. Let us compare the two ways for the linear heat and convection problems.
\begin{itemize}
  \item For the heat equation, when applying the discrete Fourier transform on the space variables, one obtains
  \[\frac{\d}{\d t} \bb{u}(t) =  A \bb{u}(t) , \qquad A = -(\bb{P}_1^2 + \cdots + \bb{P}_d^2) + \bb{V} .   \]
  It is easy to find that \eqref{heatww} coincides with \eqref{heatw} since $H_1 = A$ and $H_2 = O$. That is, the two treatments are equivalent due to the fact that $A$ is a Hermitian matrix.
  \item For the convection equation, the discrete Fourier transform on the space variables gives
\[\frac{\d}{\d t} \bb{u}(t) - \i \sum\limits_{l=1}^d \bb{P}_l \bb{u}(t) = \bb{0}, \]
which corresponds to a special case $H_1 = O$. In this case, there is no need to apply the warped phase transformation since it is already a Hamiltonian system. Of course, one can still use the $\sin(p)$-transform to get a Schr\"odinger-type system.
\end{itemize}

In summary, this general method is exactly the Schr\"odingerisation, by simply combining what we did in the previous section. $H_2$ can either be handled directly, or using Schr\"odingerisation~-~the $\sin(p)$ idea in the previous section when $H_1 = O$.

\subsection{Parity-dilating unitarisation of linear ODEs}

In addition to transforming ODEs into Hamiltonian systems, one can also translate the evolution operator into products of unitary operators that are suitable for quantum simulation. Here we present and extend the unitary dilation technique (for instance in \cite{Javier2022optionprice} applied to the Black-Scholes equation) to more general cases. In the following we still consider the ODE system \eqref{ODElinear} and assume that $\bb{u}(t)$ has been encoded as a state vector $\ket{\psi(t)}$.

The solution of $|\psi(t)\rangle$ can be approximated by
\begin{align}
\ket{\psi(t)} = \e^{(H_1 + \i H_2) t}|\psi(t_0)\rangle \approx \prod_{j=1}^{N_t} (\e^{H_1 \Delta t} \e^{\i H_2\Delta t})_j \ket{\psi(t_0)},
\end{align}
with $N_t \Delta t = t$. Higher-order terms with commutation relations between $H_1$ and $H_2$ can be added later for higher accuracy. We can enlarge the matrix (or operator) to make unitaries out of each $\e^{H_1\Delta t}\e^{\i H_2 \Delta t}$ term. Since $\e^{H_1\Delta t}=: H_{\Delta t}$ is Hermitian, one can define a unitary dilation operator $\tilde{U}$ as
\begin{equation}\label{Udilation}
\tilde{U} := \begin{bmatrix}
    H_{\Delta t} & \sqrt{I-H_{\Delta t}^2} \\
    \sqrt{I-H_{\Delta t}^2} & -H_{\Delta t}
    \end{bmatrix}=(\sigma_z \otimes I)\e^{\i\sigma_y \otimes \arccos(H_{\Delta t})},
\end{equation}
which is a particular instance of the blocking-encoding of Hermitian matrices \cite{gilyen2019quantum}, where $I$ is the identity operator and
\[\sigma_y = \begin{bmatrix} 0 & -\i \\ \i & 0 \end{bmatrix}, \qquad \sigma_z = \begin{bmatrix} 1 & 0 \\ 0 & -1 \end{bmatrix}.\]
Note that when $H_1$ is negative semi-definite,  $\|H_{\Delta t}\| \le 1$ holds, which guarantees the well-definedness of $\tilde{U}$.

\begin{remark}
  The equal sign in \eqref{Udilation} can be understood as follows. Let
\[ A(\lambda) = \begin{bmatrix}
    \lambda & \sqrt{1-\lambda^2} \\
    \sqrt{1-\lambda^2} & -\lambda
    \end{bmatrix} = \begin{bmatrix}
    \cos t & \sin t \\
    \sin t & -\cos t
    \end{bmatrix} =: B(t), \qquad  \lambda = \cos t \in [-1,1].\]
Then the column vectors of $B(t)$ are the solutions of
\[\begin{cases}
\frac{\d u_1(t)}{\d t} = -u_2(t), \\
\frac{\d u_2(t)}{\d t} = u_1(t),
\end{cases} \]
whose evolution is given by
\[
B(t) = \exp \left(\begin{bmatrix}
    0 & -1 \\
    1 & 0
    \end{bmatrix}  t \right) B(0)  =   \e^{-\i \sigma_y t} B(0) = \e^{-\i \sigma_y t} \sigma_z
\quad \mbox{or} \quad
B (t) =   \sigma_z \e^{\i \sigma_y t}.\]
This leads to $ A(\lambda) =  \sigma_z \e^{\i \sigma_y \arccos \lambda}$ and hence the desired equality by letting $\lambda = H_{\Delta t}$.
\end{remark}

\bigskip

Let $\ket{0} = [1,0]^T$. Then we can define a new unitary operator $U$ acting on $\ket{0} \otimes \ket{\psi(t_0)} = [\ket{\psi(t_0)}; \bb{0}]$ as
\begin{align}
    U \begin{bmatrix} |\psi(t_0)\rangle \\
    0
    \end{bmatrix}
: = \tilde{U} (I \otimes \e^{\i H_2\Delta t})\begin{bmatrix} \ket{\psi(t_0)} \\
    \bb{0}
    \end{bmatrix}=\begin{bmatrix} H_{\Delta t} \e^{\i H_2 \Delta t} \ket{\psi(t_0)}\\
    \sqrt{I-H_{\Delta t}^2}\e^{\i H_2\Delta t} \ket{\psi(t_0)}\end{bmatrix}, \label{evolutionaryUdilation}
\end{align}
where $\bb{0}$ has the same size with $\ket{\psi(t_0)}$. For convenience, we refer the new operator $U$ as the (first-order) evolutionary unitary dilation operator.
Now if proceed to use this state directly to the next time step and multiply by another $U$, one does not recover a term $(H_{\Delta t} \e^{\i H_2 \Delta t})^2 \ket{\psi(t_0)}$ in the top entry, but rather it is augmented by another term $(\sqrt{I-H_{\Delta t}^2}\e^{\i H_2\Delta t} )^2 \ket{\psi(t_0)}$. This means we cannot post-select at the end, but only at every time step. For instance, in the second step the quantum state should be post-selected as
\[\begin{bmatrix} H_{\Delta t} \e^{\i H_2 \Delta t} \ket{\psi(t_0)}\\
    \bb{0}\end{bmatrix}
= \begin{bmatrix}  \ket{\psi(t_1)}\\
    \bb{0}\end{bmatrix} = \ket{0} \otimes \ket{\psi(t_1)},\]
and hence we can repeat the procedure. However, if the success probability of every time step to obtain the top entry is $P$, then after $N_t$ time-steps, a single success requires $\mathcal{O}(P^{-N_t})$ copies of $|\psi(t_0)\rangle$ at the beginning. A simple case is that $H_1$ commutes with $H_2$ as for the Black-Scholes equation \cite{{Javier2022optionprice}}. In this situation,
\[\e^{A t} \ket{\psi(t_0)} = \e^{\i H_2 t} \e^{H_1 t} \ket{\psi(t_0)},\]
and one can first apply the unitary dilation technique to $\e^{H_1 t}$ and perform one post-selection at the end of the simulation of $H_1$.

However, $H_1$ and $H_2$ do not commute in general. We go for an alternative approach, where we instead begin with the state $|0\rangle^{\otimes N_t^*} \otimes \ket{\psi(t_0)}$, where $N_t^* = \log N_t$. In this case we only need to perform one post-selection at the very end. Observing that
\[\begin{bmatrix}  \ket{\psi(t_0)}\\
    \bb{0}\end{bmatrix} \xrightarrow{U} \begin{bmatrix}  \ket{\psi(t_1)}\\
    * \end{bmatrix},\]
we can introduce the following translations
\[\begin{bmatrix}  \ket{\psi(t_0)}\\
\bb{0} \\
\bb{0} \\
\vdots \\
\bb{0}  \\
\bb{0}
\end{bmatrix} \xrightarrow{U_1} \begin{bmatrix}  \ket{\psi(t_1)}\\
* \\
\bb{0} \\
\vdots \\
\bb{0}  \\
\bb{0}
\end{bmatrix}
\xrightarrow{U_2} \begin{bmatrix}  \ket{\psi(t_2)}\\
* \\
* \\
\vdots \\
\bb{0}  \\
\bb{0}
\end{bmatrix} \xrightarrow{U_3} \cdots ,\]
where the vector has $N_t$ zero vectors, and each unitary $U_j$ is basically the evolutionary unitary dilation operator $U$ in Eq.~\eqref{evolutionaryUdilation}, but applied to different {\it pairs} of qubit systems.  For this reason, we refer to it as the parity-dilating unitarisation approach in this article. The costs are identical for all $j$, and the explicit formula of $U_j$ can be written as
\begin{align*}
U_j
& := |0\rangle \langle 0| \otimes \e^{\i H_2 \Delta t}H_{\Delta t}+(|j\rangle \langle 0|+|0\rangle \langle j|)\otimes \e^{\i H_2 \Delta t}\sqrt{I-H_{\Delta t}^2} \\
& \qquad -|j\rangle \langle j|\otimes \e^{\i H_2 \Delta t}H_{\Delta t}+\sum_{k \neq 0, j}^{N_t-1}|k\rangle \langle k|\otimes I.
\end{align*}
This means if assume access to $U_j$ (as an oracle), then one needs to apply this oracle $N_t$ times to the state $|0\rangle^{\otimes N_t^*} \otimes |\psi(t_0)\rangle$ and then just to perform a single post-selection at the end. However, we note that this assumption is {\it not} equivalent to the usual assumption of unitary decomposition, where we decompose into single qubit gates and CNOTs or some other combination.

\begin{theorem}\label{thm:unitarisation}
Assume sparse access to $\arccos(H_1 \Delta t)$ with sufficient precision and suppose that the time step satisfies $\|A\|_1 \Delta t \le 1$. Then the first-order parity-dilation unitarisation method has gate complexity
\begin{align*}
N_{\text{Gates,Unitary}}
 = \log N_A \cdot \widetilde{\mathcal{O}}\Big( T/\Delta t \cdot  s(\arccos(H_1 \Delta t)) +  s(A) \|A\|_{\max} \Big) .
\end{align*}
In particular, if $H_1$ can be diagonalised in the momentum basis and $H_2$ is a diagonal matrix,  and $\|H_1\| \Delta t \le 1$, then the parity-dilation unitarisation method can be simulated with
\[N_{\text{Gates,Unitary}} = T/\Delta t \cdot \mathcal{O}(d m \log m).\]
\end{theorem}
\begin{proof}
1) The first-order evolutionary unitary dilation operator is
\[U = (\sigma_z \otimes I) \e^{\i\sigma_y \otimes \arccos(H_1 \Delta t)}  (I \otimes \e^{\i H_2\Delta t}).\]
By Lemma \ref{lem:complexityHamiltonian}, the operators $\e^{\i\sigma_y \otimes \arccos(H_1 \Delta t)}$ and $\e^{\i H_2\Delta t}$ can be simulated with
\[
  \widetilde{\mathcal{O}}( m_d  s(\arccos(H_1 \Delta t)) \|\arccos(H_1 \Delta t)\|_{\max} )
+  \widetilde{\mathcal{O}}(\Delta t m_d s(H_2) \|H_2\|_{\max} )
\]
2-qubits gates (note that $\Delta t$ is not a factor in the first term), where  $\arccos(H_1 \Delta t)$ is defined according to the Taylor expansion
\[\arccos x = \frac{1}{2} \pi - x - \frac{1}{6} x^3 - \frac{3}{40}x^5 - \frac{5}{112}x^7 - \frac{35}{1152}x^9 - \cdots,\]
which is well-defined when $|x|\le 1$. One easily finds that $\arccos(H_1 \Delta t)$ is well-defined when the time step satisfies $\|H_1\|_1 \Delta t \le 1$. In fact, since $H_1$ is Hermitian, there exists a unitary matrix $V$ such that $V^{-1} H_1 V = \Lambda_1$, where $\Lambda_1$ is the diagonal matrix consisting of the eigenvalues of $H_1$. By definition,
\begin{equation}\label{diagonalisation}
\arccos(H_1 \Delta t) = T \arccos(\Lambda_1 \Delta t) T^{-1},
\end{equation}
where $\arccos(\Lambda_1 \Delta t)$ is obviously well-defined due to the fact that $|\lambda(H_1) |_{\max} \Delta t \le \|H_1\|_1 \Delta t \le \|A\|_1 \Delta t \le 1$. Therefore,
\begin{align*}
\|\arccos(H_1 \Delta t)\|_{\max}
& \le \|\arccos(H_1 \Delta t)\|_1  \le \frac{1}{2} \pi + \|H_1 \Delta t\|_1 + \frac{1}{6} \|H_1 \Delta t\|_1^3 + \cdots \\
& \le \frac{1}{2} \pi + 1 + \frac{1}{6} + \frac{3}{40} + \cdots  = \arccos(-1) = \pi.
\end{align*}

2) For the special case, the first-order evolutionary unitary dilation operator is
\[
U=(\sigma_z \otimes I) (I \otimes \Phi^{\otimes^d} ) \e^{\i \sigma_y \otimes \arccos (\Lambda_1 \Delta t)}(I \otimes (\Phi^{\otimes^d} )^{-1} ) ( I  \otimes \e^{\i H_2 \Delta t} ) ,
\]
where $\Lambda_1$ and $H_2$ are diagonal matrices, and $\arccos (\Lambda_1 \Delta t)$ is well-defined if $\|H_1\| \Delta t \le 1$. The gate complexity is obviously given by
\begin{align*}
N_{\text{Gates,Unitary}}
  = n (2 \mathcal{O}(d m \log m)  + 2 \mathcal{O}(d m) ) = t/\Delta t \cdot \mathcal{O}(d m \log m).
\end{align*}
This completes the proof.
\end{proof}

%\begin{remark} \label{rem:sarc} \color{blue}
%It is simple to find that
%\[s(A+B) \le s(A) + s(B) \qquad  \mbox{and} \qquad s(A^k) \le 2^{k-1} s(A).\]
%This leads to
%\[s(\arccos (H_1\Delta t) ) = 1 + s(H_1 \Delta t)(2^2 + 2^4 + 2^6 + \cdots) \le 1 + s(A)(2^3 + 2^5 + 2^7 + \cdots),\]
%which implies that the higher order truncation gives a more dense matrix in general.
%\end{remark}

\begin{remark}
Note that for $\arccos(H_1 \Delta t)$ to be well-defined, one needs $\|H_1\| \Delta t \le 1$. If $H_1$ is the discrete heat operator, then $\|H_1\|= \mathcal{O}(\Delta x^2)$ in one dimension, hence one needs $\Delta t = \mathcal{O}(\Delta x^2)$, which is the CFL stability condition for solving the heat equation using explicit time discretisation. On the other hand, the time splitting method is unconditionally stable, thus one can take $\Delta t = \mathcal{O}(\Delta x)$ there.
According to the special case in Theorems \ref{thm:Schrodingerisation} and \ref{thm:unitarisation}, the heat equation \eqref{heateq} can be simulated with
\begin{align*}
& N_{\text{Gates,Schr}} =  T/\Delta t \cdot \mathcal{O}( d m \log m + m_p \log m_p) = T/ \Delta x\cdot \mathcal{O}( d m \log m + m_p \log m_p),\\
& N_{\text{Gates,Unitary}} = T/\Delta t \cdot \mathcal{O}(d m \log m) = T /\Delta x^2 \cdot \mathcal{O}(d m \log m).
\end{align*}
The mesh strategy is $d\Delta x^\ell \sim \varepsilon$ and $\Delta p \sim \varepsilon$, which gives
\[\frac{N_{\text{Gates,Schr}}}{N_{\text{Gates,Unitary}}} = \Delta x \cdot \mathcal{O}\Big(  1 + \frac{\ell}{d} \frac{\log (1/\varepsilon)}{\log(d/\varepsilon)} \Big).\]
This implies that the parity-dilating unitarisation method requires more computational cost if $\ell$ is not large, that is, the solution $u$ is not sufficiently smooth.
\end{remark}

%As pointed out before, the parity-dilating unitarisation method usually takes more time for the special cases where one of the matrices can be diagonalised in the momentum basis and the other one is diagonal. {\color{red} It should be pointed out that the cost of our method may be slightly overhead than the unitarisation method when the two operators can be diagonalised in the same basis as observed in the simulation of the Black-Scholes equation (see Subsection \ref{subsect:BlackScholes})}.

In the following, we focus on distinguishing the parity-dilating unitarisation method with our Schr\"odingerisation one with respect to the sparsity.

For spectral discretisations, the coefficient matrix arising from the differential operators is always dense in the original variables due to the existence of the DFT matrix. It may be hard to get a sparse system when the equations have varying coefficients. For a dense matrix, the preparation of the matrix $\arccos(H_1 \Delta t)$ seems to be rather involved. A candidate for the implementation is the Taylor expansion, which, however, cannot resolve the sparsity problem.
%and may destroy lots of important physical properties of the original operator. It also makes the quantum implementation more difficult and probably with added extra cost.

The finite difference discretisation usually yields sparse systems but with low accuracy than the spectral method. The matrix $\arccos(H_1 \Delta t)$ is often dense even if $H_1$ is very sparse. This is like the transformation in DFT, which makes the diagonal matrix a dense one (see \eqref{diagonalisation}).  This means $s(\arccos(H_1 \Delta t))$ usually scales as $\mathcal{O}(\Delta x^{-d})$, which can be far greater than $s(H_1) \|H_1\|_{\max}/(\Delta p)$ when $d$ is large (the later one may scale as $\mathcal{O}(d^2/(\Delta x^\alpha \Delta p))$). On the other hand, the density makes the implementation of $\arccos(H_1 \Delta t)$ difficult.
 It should be pointed out that for some cases, for instance, the central difference of the Laplacian, the matrix can be diagonalised in the discrete Fourier, discrete sine or discrete cosine transformation matrix. In this case, the $\arccos(H_1 \Delta t)$ can be efficiently implemented since the transformation matrices can be realised by the fast Fourier transform.

\subsection{Block-encoding unitarisation of linear ODEs}

The parity-dilation method shows a particular example of a block-encoding strategy. More general results can be derived. In Section 4.1 of \cite{An2022blockEncodingODE},  the authors presented the block-encoding technique to solve \eqref{ODElinear} with negative definite $A$ and time-independent $\bb{b}$, where $\|A\| \le 1$ is assumed for technical simplicity.  Based on the following evolutionary form
\[\bb{u}(t) = \e^{A t} \bb{u}(0)  +  \int_0^T \e^{A(T-s)} \d s\bb{b} = \e^{A t} \bb{u}(0) + (\e^{AT} - I) A^{-1} \bb{b},\]
the algorithm first separately computes the homogeneous and the inhomogeneous parts and then combines them together using the technique of linear combination of quantum states. For construction of the linear combination, we refer the reader to Lemma 22 there when given the block-encodings of $\e^{A t}$ and $\int_0^T \e^{A(T-s)} \d s$.

The query complexity of the block-encoding technique in \cite{An2022blockEncodingODE} is based on access to the block-encoding of matrix $A$, which can be difficult to realise. If we assume access to the block-encoding of $A$, then the complexity is dependent of the condition number of $A$ (i.e., the inverse of $\delta$ in Theorem 23), which could give rise to a large factor. This is not the case for the Hamiltonian simulation based algorithms. However, if access to the block-encoding of $A^2$ is assumed, this factor is no longer present.

\subsection{Imaginary time evolution methods for linear ODEs}
Imaginary time evolution methods refers to the application of a Wick rotation $\tau = -\i t$ to transform the heat equation $\partial_t u = \partial_{xx}u$ into a Schr\"odinger equation $-\i \partial_\tau u(\tau, x) = \partial_{xx} u(\tau,x)$. This has been applied  in heuristic schemes for quantum problems for instance in \cite{motta2020determining, mcardle2019variational, seki2021quantum}. In this case, the state obeying Schr\"odinger's equation and its imaginary time counterpart do not have the same evolution. This means that, except for  the steady state solution, extra resources are necessary to map between the solution in the unitarily evolving system to the other. Usually some heuristic techniques are used, like variational methods or classical optimisation that requires input from quantum measurements on the quantum states themselves. These, however, have the benefit of being implementable on hybrid classical-quantum devices.

The imaginary time evolution approach can also be used to deal with the general linear ODE system \eqref{ODElinear} in principle. However, it may change the nature of the underlying PDEs. For example, if it is done for the heat equation, then one changes from  a dissipative equation~-~which has the merit of converging fast (exponentially) to the ground state~-to a Hamiltonian system that converges slowly to the ground or steady state.

This differs from our method, which does not rely on heuristic methods. Here we instead transform the heat equation to phase space by using the warped phase transformation. We can also apply this to ground state or steady state preparation \cite{Schrshort}.

\section{Applications to more PDEs} \label{sec:applications}

In this section, we apply our method to more typical examples. For simplicity, we only provide the time complexity in terms of the number of qubits since the error estimates may be rather involved.

\subsection{The Black-Scholes equation}\label{subsect:BlackScholes}

The Black-Scholes equation
\[\frac{\partial{V}}{\partial t} + r S \frac{\partial{V}}{\partial S} + \frac12 \sigma^2S^2 \frac{\partial{V}^2}{\partial S^2} = rV,\]
is a PDE that evaluates the price of a financial derivative, where $r$ and $\sigma$ are constants. For a specific derivative contract, the problem is to determine its present price $V(t=0,S)$ according to the terminal price $V(t = T, S )$ of the option \cite{Javier2022optionprice}. The change of variables $S = \e^x$, $-\infty < x < \infty$ leads to a backward parabolic equation
\[\frac{\partial{V}}{\partial t} + (r-\frac{\sigma^2}{2})\frac{\partial{V}}{\partial x} + \frac{\sigma^2}{2} \frac{\partial{V}^2}{\partial x^2} = rV.\]
One can reverse time $t \to \tau = T-t$ to get a forward parabolic equation
\begin{equation}\label{BS}
\frac{\partial{V}}{\partial \tau} = (r-\frac{\sigma^2}{2})\frac{\partial{V}}{\partial x} + \frac{\sigma^2}{2} \frac{\partial{V}^2}{\partial x^2} - rV.
\end{equation}
This is a typical example, in which the underlying operators can be diagonalised in the momentum basis, that has been resolved by the unitarisation approach in \cite{Javier2022optionprice}.

We first consider the Schr\"odingerisation method. By the warped phase transformation $W(t,x,p) = \e^{-p}V(t,x)$ with periodic extension of the initial data, one gets
\[\partial_{\tau} W = (r-\frac{\sigma^2}{2})\partial_x W + (\frac{\sigma^2}{2} \partial_{xx} - rI)(- \partial_p W).\]
By repeating the previous calculations,  it is straightforward to derive a Hamiltonian system
\[\frac{\d}{\d t} \tilde{\bb{W}}(t) = \i \bb{H} \tilde{\bb{W}}(t), \qquad \bb{H} = (r-\frac{\sigma^2}{2})(P_\mu \otimes I) - (\frac{\sigma^2}{2}P_\mu^2 + rI) \otimes D_\mu,\]
where $\tilde{\bb{W}} = (I \otimes F_p^{-1}) \bb{W}$ and $\bb{W}(t) = \sum_{j, k} W(t,x_j,p_k) \ket{j,k}$.
Performing the change of variables $\tilde{\bb{c}} = (F_x^{-1} \otimes I) \tilde{\bb{W}}$, where $F_x = \Phi$, one has
\[\frac{\d}{\d t} \tilde{\bb{c}}(t) = \i \tilde{\bb{H}} \tilde{\bb{c}}(t), \qquad \tilde{\bb{H}} = (r-\frac{\sigma^2}{2})(D_\mu \otimes I) - (\frac{\sigma^2}{2}D_\mu^2 + rI) \otimes D_\mu.\]

\begin{theorem}
The Schr\"odingerisation approach for the Black-Scholes equation can be simulated with gate complexity given by
\[N_{\text{Gates,Schr}} = \mathcal{O}( m \log m + m_p \log m_p).\]
\end{theorem}
\begin{proof}
The diagonal unitary operator $\e^{-\i \tilde{\bb{H}} t}$ can be simulated with
 \[N_{\text{Gates}}(\e^{\i \tilde{\bb{H}} t}) = \mathcal{O}(m_H) = \mathcal{O}(m + m_p) \]
2-qubits gates  (see Algorithm I in Remark \ref{rem:algI}). The result follows by adding the number of gates for the QFT, which is only performed twice for both $x$ and $p$.
\end{proof}

For the unitarisation method, we rewrite the equation \eqref{BS} as
\[\partial_\tau V = \hat{H}_{\text{BS}} V,\]
where $\hat{H}_{\text{BS}} = \hat{H}_1 + \hat{H}_2$, with
\[\hat{H}_1 V = - \i (r-\frac{\sigma^2}{2})\partial_x V \qquad \mbox{and} \qquad \hat{H}_2 V =  \frac{\sigma^2}{2} \partial_{xx} V - rV\]
representing the Hermitian and non-Hermitian parts, respectively. Using the DFT, the above equation can be written as
\[\partial_\tau \bb{V} = (\bb{H}_1 + \i  \bb{H}_2) \bb{V},\]
with
\[ \bb{V}(t) = \sum\limits_j V(t,x_j) \ket{j}, \qquad \bb{H}_1 = (r-\frac{\sigma^2}{2})P_\mu,  \qquad \bb{H}_2 =  - (\frac{\sigma^2}{2}P_\mu^2 + rI). \]
Let $\tilde{\bb{V}} = F_x^{-1} \bb{V}$. The equation can be rewritten as the diagonalised form in the momentum space:
\[\partial_\tau \tilde{\bb{V}} = (\tilde{\bb{H}}_1 + \i  \tilde{\bb{H}}_2) \tilde{\bb{V}}, \qquad
\tilde{\bb{H}}_1 = (r-\frac{\sigma^2}{2})D_\mu,  \qquad \tilde{\bb{H}}_2 =  - (\frac{\sigma^2}{2}D_\mu^2 + rI).\]

\begin{theorem}\label{thm:unitarisationBS}
The parity-dilating unitarisation approach for the Black-Scholes equation can be simulated with gate complexity given by
\[N_{\text{Gates,Unitary}} = \mathcal{O}( m \log m).\]
\end{theorem}
\begin{proof}
Since $\tilde{\bb{H}}_1$ commutes with $\tilde{\bb{H}}_2$, the evolution of $\tilde{\bb{V}}$ can be simply written as
\[\tilde{\bb{V}}(\tau) = \e^{(\tilde{\bb{H}}_1 + \i  \tilde{\bb{H}}_2)\tau }\tilde{\bb{V}}(0)
= \e^{\tilde{\bb{H}}_1 \tau}  \e^{\i  \tilde{\bb{H}}_2 \tau }\tilde{\bb{V}}(0).\]
Thus one just needs to perform the unitary dilation technique once.  According to the proof of Theorem \ref{thm:unitarisation}, the evolutionary unitary dilation operator is
\[U = (\sigma_z \otimes \mathrm{I}) \e^{\i\sigma_y \otimes \arccos(\tilde{\bb{H}}_1 \tau)}  (\mathrm{I} \otimes \e^{\i \tilde{\bb{H}}_2 \tau}),\]
where the diagonal unitary operators on the right-hand side can be simulated with $\mathcal{O}(m+1)$ gates. The result follows by adding the number of gates for the QFT.
\end{proof}

For this simple example, there is a slight overhead in time complexity for the  Schr\"odingersation approach, i.e.,
\[N_{\text{Gates,Schr}} - N_{\text{Gates,Unitary}} = \mathcal{O}(m_p \log m_p) = \mathcal{O}(\log (1/\varepsilon))\]
since $m_p \sim \log(1/\Delta p)$ and $\Delta p \sim \varepsilon$.
Such a conclusion is obviously valid for cases where the underlying operators can be diagonalised in the same basis (in this case, the arccos is not an issue as observed in the proof of Theorem \ref{thm:unitarisationBS}), which are often encountered for the differential operators with constant coefficients.

\subsection{The Fokker-Planck equation}

The Fokker-Planck equation describes the time evolution of the probability density function $f(t,x)$ of the velocity of a particle under the influence of drag forces and random forces \cite{Pav}. It has the form
\begin{equation}\label{FP0}
    \partial_t f = -\nabla \cdot ({\nabla V(x) f}) + \sigma \Delta f,
\end{equation}
where $V(x)$ is a scalar function and $\sigma>0$ is a constant. The first term on the right-hand side is called the drifted term, and the second term is the diffusion term generated by white noise. This equation has the steady state solution $f=\e^{-V(x)/\sigma}$. For convenience, we assume the periodic boundary conditions with $x =(x_1,\cdots,x_d)\in [-1,1]^d$.

\subsubsection{The conservation form}

The equation \eqref{FP0} can also be written as
\begin{equation}\label{FP}
    \partial_t f = \sigma \nabla \cdot \Big(\e^{-V/\sigma} \nabla \Big(\e^{V/\sigma}  f \Big) \Big).
\end{equation}
As done for the heat equation, one can introduce the transformation $F(t,x,p) = \e^{-p} f(t,x)$ and extend the initial data to $p<0$ to obtain
\[\begin{cases}
\partial_t F = \sigma \nabla_x \cdot \Big(\e^{-V/\sigma} \nabla_x \Big(\e^{V/\sigma}  (- F_p) \Big) \Big) \\
F(0,x,p) = \e^{-|p|} f(t,x).
\end{cases}\]
Apply the discrete Fourier transformation on both $x$ and $p$ to get
\begin{align*}
& \nabla_x \cdot \Big(\e^{-V/\sigma} \nabla_x \Big(\e^{V/\sigma}  (- F_p) \Big) \Big)
  = \sum\limits_{l=1}^d \partial_{x_l} \Big(\e^{-V/\sigma} \partial_{x_l} \Big(\e^{V/\sigma}  (- F_p) \Big) \Big) \\
& = -\i \sum\limits_{l=1}^d (-\i \partial_{x_l}) \Big(\e^{-V/\sigma} (-\i \partial_{x_l}) \Big(\e^{V/\sigma}  (- (-\i \partial_p)) F \Big) \Big) \\
& \longrightarrow \i \sum\limits_{l=1}^d  (\bb{P}_l \otimes I)(\bb{e}^{-\bb{V}/\sigma} \otimes I) (\bb{P}_l \otimes I)(\bb{e}^{\bb{V}/\sigma} \otimes I) (I^{\otimes^d} \otimes P_\mu) \bb{F} \\
& = \i \sum\limits_{l=1}^d  (\bb{P}_l \bb{e}^{-\bb{V}/\sigma} \bb{P}_l \bb{e}^{\bb{V}/\sigma} \otimes  P_\mu) \bb{F},
\end{align*}
where $\bb{P}_l$ is the same as the one given for the heat equation, and $\bb{e}^{\bb{V}} = \text{diag}(\bb{g})$ is a diagonal matrix with the diagonal vector given by
$\bb{g} = \sum_{\bb{j}} \e^{V(x_{\bb{j}})} \ket{\bb{j}}$.
It should be pointed out that $\bb{e}^{\bb{V}}$ is not the matrix exponential $\e^{\bb{V}}$, where we have used the bold symbol $\bb{e}$ to indicate the difference.

Let $\bb{A}_l = \bb{P}_l \bb{e}^{-\bb{V}/\sigma} \bb{P}_l$, which are Hermitian matrices. One has the following ODEs
\[\frac{\d }{\d t} \bb{F} = \i \sigma \sum\limits_{l=1}^d  (\bb{A}_l \bb{e}^{\bb{V}/\sigma} \otimes P_\mu) \bb{F}.\]
Defining $\tilde{\bb{F}} = (\bb{e}^{{\bb{V}}/(2\sigma)} \otimes F_p^{-1})\bb{F}$, we get a Hamiltonian system
\[\frac{\d }{\d t} \tilde{\bb{F}} = \i \bb{H}_1 \tilde{\bb{F}},\]
where
\[\bb{H}_1 =  ( \bb{B}_1 + \cdots + \bb{B}_d)\otimes D_\mu,  \]
\[\bb{B}_l = \sigma \bb{e}^{{\bb{V}}/(2\sigma)} \bb{A}_l \bb{e}^{{\bb{V}}/(2\sigma)}, \quad \bb{A}_l = \bb{P}_l \bb{e}^{-\bb{V}/\sigma} \bb{P}_l.\]

\subsubsection{The heat equation form}

Using the transformation $\psi(t,x)=\e^{V/(2\sigma)}f $ (recall the definition of  $\tilde{\bb{F}}$ in the conservation form), one easily gets $\psi$ satisfies the imaginary time Schr\"odinger or heat equation \cite{MarKVill}
\begin{equation}\label{FPS}
    \partial_t \psi = \sigma \Delta \psi - U(x)\psi
\end{equation}
where
\[U(x):= \frac{|\nabla V|^2}{4\sigma}-\frac{1}{2} \Delta V.\]
Since the equation \eqref{FPS} has the same form of the heat equation in \eqref{heateq}, one can introduce the transformation technique to get
\[\frac{\d}{\d t} \tilde{\bb{\Psi}}(t) = \i \bb{H}_2  \tilde{\bb{\Psi}}(t) ,\]
where
\[\bb{H}_2 = ( \sigma(\bb{P}_1^2 + \cdots + \bb{P}_d^2) + \bb{U} ) \otimes D_\mu  \]
is a Hermitian matrix and $\bb{U}$ is defined as $\bb{V}$.

\begin{remark}
Obviously, it is more time consuming to perform the quantum simulation of the first matrix $\bb{H}_1$. In fact, $\bb{B}_l$ are not sparse in $x_l$ direction, hence the simulation of $\bb{H}_1$ is not sparse along $x_l$ direction. However, $\bb{P}_l^2$ in $\bb{H}_2$ can be efficiently implemented in the frequency space by using the quantum Fourier transform.
\end{remark}

\subsection{The linear Boltzmann equation}

In this section we consider the linear Boltzmann equation with isotropic scattering
\cite{Chand}
\[\begin{cases}
\partial_t f + \xi \cdot \nabla_x f = \frac{1}{|\Omega|} \int_{\Omega} f(t,x,\xi') \d \xi' - f, \\
f(0,x,\xi) = f_0(x,\xi),
\end{cases}\]
where $f = f(t,x,\xi)$, $x =(x_1,\cdots,x_d)\in [-1,1]^d$,  and $\xi = (\xi_1,\cdots,\xi_{d-1})$ is a vector on the unit sphere $\mathbb{S}^{d-1}$ in $\mathbb{R}^d$. We assume the periodic boundary conditions are imposed.

\subsubsection{The Hamiltonian system}

Proceeding as in the previous sections, we introduce the warped phase transformation
\[F(t,x,\xi,p) = \e^{-p} f(t,x,\xi), \qquad p>0\]
with the initial data symmetrically extended to $p<0$, and find that $F$ solves
\[\begin{cases}
\partial_t F + \xi \cdot \nabla_x F = -\Big( \frac{1}{|\Omega|} \int_{\Omega} \partial_p F(t,x,\xi',p) \d \xi' - \partial_p F \Big), \\
F(0,x,\xi,p) = F_0(x,\xi,p):= \e^{-|p|}f_0(x,\xi).
\end{cases}\]

We use the discrete-ordinate method to discretise the integral. Let $(w_k,\xi_k)$ be the quadrature weights and points, where $\xi_k= (\xi_{k1},\cdots,\xi_{kd})$ and $1\le k \le N$. One has the semi-discrete system:
\[\partial_t F_m + \xi_m \cdot \nabla_x F_m = - \Big( \sum\limits_k w_k \partial_p F_k - \partial_p F_m \Big), \qquad m = 1,2,\cdots,N,\]
where $F_m = F_m(t,x,p) := F(t,x,\xi_m, p)$.

Let
\[\bb{F}_m = \sum\limits_{\bb{j},j_p }\bb{F}_m(t,x_{\bb{j}}, p_{j_p}) \ket{\bb{j}, j_p}, \qquad \bb{j} = (j_1,\cdots,j_d), \]
which is a column vector of $M^{d+1}$ entries.
Taking the discrete Fourier transformation both on $x$ and $p$ yields
\[\frac{\d }{\d t} \bb{F}_m + \i \sum\limits_{l=1}^d \xi_{ml} (\bb{P}_l \otimes I)\bb{F}_m = -\i (I^{\otimes^d} \otimes P_\mu) \Big( \sum\limits_{k=1}^N w_k \bb{F}_k - \bb{F}_m \Big), \qquad m = 1,\cdots, N,\]
where $\bb{P}_l = I^{\otimes^{l-1}} \otimes P_\mu \otimes I^{\otimes^{d-l}}$.
The above system can be rewritten as
\begin{equation} \label{BoltzmannF}
\frac{\d }{\d t} \bb{F} + \i \sum\limits_{l=1}^d ( \Lambda_{\xi_l} \otimes \bb{P}_l \otimes I ) \bb{F}
= -\i (W-I) \otimes I^{\otimes^d} \otimes P_\mu \bb{F},
\end{equation}
where $\bb{F} = [\bb{F}_1; \cdots; \bb{F}_N]$ with ``;" indicating the straightening of $\{\bb{F}_i\}_{i\ge 1}$ into a column vector, $\Lambda_{\xi_l} = \text{diag}(\xi_{1l},\xi_{2l},\cdots, \xi_{Nl})$, and
\[W = \Xi \Lambda_w, \quad \Xi = (\Xi_{ij})_{N\times N}, ~\Xi_{ij} \equiv 1, \quad \Lambda_w = \text{diag}(w_1,w_2,\cdots,w_N).\]

Define $\tilde{\bb{F}} = ( \Lambda_w^{1/2} \otimes I^{\otimes^d} \otimes  F_p^{-1} ) \bb{F} $. One has from \eqref{BoltzmannF} that
\begin{align*}
\frac{\d }{\d t} \tilde{\bb{F}} + \i \sum\limits_{l=1}^d ( \Lambda_w^{1/2} \Lambda_{\xi_l}\Lambda_w^{-1/2} \otimes \bb{P}_l \otimes I ) \tilde{\bb{F}}
= -\i (\Lambda_w^{1/2} \Xi \Lambda_w^{1/2}-I) \otimes I^{\otimes^d} \otimes D_\mu) \tilde{\bb{F}}.
\end{align*}
Noting that $\Lambda_w^{1/2} \Lambda_{\xi_l}\Lambda_w^{-1/2} = \Lambda_{\xi_l}$, we finally derive a Hamiltonian system
\begin{equation}\label{BoltzmannHamiltonian}
\frac{\d }{\d t} \tilde{\bb{F}} = -\i \bb{H} \tilde{\bb{F}},
\end{equation}
where
\[\bb{H} =  \sum\limits_{l=1}^d ( \Lambda_{\xi_l}\otimes \bb{P}_l \otimes I ) + (\Lambda_w^{1/2} \Xi \Lambda_w^{1/2}-I) \otimes I^{\otimes^d} \otimes D_\mu\]
is a Hermitian matrix.

\subsubsection{The quantum simulation}

For the Hamiltonian system \eqref{BoltzmannHamiltonian}, one can solve
\[\frac{\d }{\d t} \tilde{\bb{F}} = -\i \bb{H}_\xi \tilde{\bb{F}}, \qquad \bb{H}_\xi = \sum\limits_{l=1}^d ( \Lambda_{\xi_l}\otimes \bb{P}_l \otimes I )\]
for one time step, followed by solving
\[\frac{\d }{\d t} \tilde{\bb{F}} = -\i \bb{H}_w \tilde{\bb{F}}, \qquad \bb{H}_w = (\Lambda_w^{1/2} \Xi \Lambda_w^{1/2}-I) \otimes I^{\otimes^d} \otimes D_\mu\]
again for one time step. By introducing $\tilde{\bb{c}} = (I^{\otimes^d} \otimes F_x^{-1} \otimes I)\tilde{\bb{F}}$, where $F_x = \Phi^{\otimes^d}$, the system in the first step is transformed into
\[\frac{\d }{\d t} \tilde{\bb{c}} = -\i \bb{H}_\xi \tilde{\bb{c}}, \qquad \bb{H}_{D,\xi} = \sum\limits_{l=1}^d ( \Lambda_{\xi_l}\otimes \bb{D}_l^\mu \otimes I ),\]
where $\bb{H}_{D,\xi}$ is a diagonal matrix.

\begin{theorem}
The solution to the Boltzmann equation can be simulated with gate complexity given by
\[N_{\text{Gates}} = \widetilde{\mathcal{O}}(m_H N^2 /\Delta p + m_H /\Delta x ) + \mathcal{O}(T/\Delta t \cdot d m \log m ) + \mathcal{O}(m_p \log m_p),\]
where $N$ is the number of quadrature points and $m_H = dm + m_p$.
\end{theorem}
\begin{proof}
Given the initial state of $\tilde{\bb{F}}^0$, applying the inverse QFT to the $x$-register, one gets $\tilde{\bb{c}}^0$.
At each time step, one needs to consider the following procedure
\[\tilde{\bb{c}}^n
\xrightarrow {\e^{-\i \bb{H}_{D,\xi} \Delta t}} \tilde{\bb{c}}^*
\xrightarrow {I^{\otimes^d} \otimes F_x\otimes I } \tilde{\bb{F}}^*
\xrightarrow {\e^{-\i \bb{H}_w \Delta t}} \tilde{\bb{F}}^{n+1}
 \xrightarrow {I^{\otimes^d} \otimes F_x^{-1}\otimes I } \tilde{\bb{c}}^{n+1}.\]

By Lemma \ref{lem:complexityHamiltonian},$\e^{-\i \bb{H}_w \Delta t}$ can be simulated with
\[N_{\text{Gates}}(\e^{-\i \bb{H}_w \Delta t}) = \widetilde{\mathcal{O}}\Big(\Delta t m_H s(\bb{H}_w)\|\bb{H}_w\|_{\max} \Big) = \widetilde{\mathcal{O}}(\Delta t \cdot m_H N^2/\Delta p ),\]
where $s(\bb{H}_w) = N $ and $\|\bb{H}_w\|_{\max}\lesssim N M_p$. Similary, $\e^{-\i \bb{H}_{D,\xi} \Delta t}$ can be simulated with
 \[N_{\text{Gates}}(\e^{-\i \bb{H}_{D,\xi} \Delta t}) = \widetilde{\mathcal{O}}\Big(\Delta t m_H s(\bb{H}_{D,\xi})\|\bb{H}_{D,\xi}\|_{\max} \Big)
 = \widetilde{\mathcal{O}}(\Delta t m_H /\Delta x ),\]
where $s(\bb{H}_{D,\xi}) = 1$ and $\|\bb{H}_{D,\xi}\|_{\max} \le M = \mathcal{O}( 1/\Delta x)$.

It is known that the quantum Fourier transforms $F_x$  can be implemented using $ \mathcal{O}(d m \log m)$ gates. Therefore, the gate complexity required to iterate to the $n$-th step for the system \eqref{BoltzmannHamiltonian} is
\begin{align*}
N_{\text{Gates}}
& = n \Big( \widetilde{\mathcal{O}}(\Delta t \cdot m_H  N^2/\Delta p + \Delta t m_H /\Delta x ) +  \mathcal{O}(d m \log m) \Big) \\
& = \widetilde{\mathcal{O}}(m_H N^2 /\Delta p + m_H /\Delta x ) + \mathcal{O}(T/\Delta t \cdot d m \log m ).
\end{align*}
The proof is complete.
\end{proof}

We remark that $\|H\|_{\max}$ in Algorithm II (see Lemma \ref{lem:complexityHamiltonian}) is not significantly amplified by the auxiliary variable $p$. This is due to the presence of the first-order derivative with respect to $x$ in the convection term. This is very different from the Schr\"odingerisation of the linear convection equation in Subsect.~\ref{subsec:convection}, where the second-order derivative with respect to $p$ is included, which leads to $1/\Delta p^2$ as the multiplicative factor in the time complexity if Algorithm II is used.

\subsection{The Vlasov-Fokker-Planck equation} \label{subsec:VFP}

We now present an example in which the first use of the warped phase transformation fails.

Consider $f=f(t,x,\xi)>0$ satisfying the Vlasov-Fokker-Planck (VFK) equation (also called the Klein-Kramers-Chandrasekhar equation)  \cite{McKean}:
\begin{equation*}
    \partial_t f+ \xi\cdot \nabla_x f -\nabla V(x) \cdot \nabla_\xi f = \nabla_\xi \cdot (\xi f + \nabla_\xi f),
\end{equation*}
where $x, \xi \in \mathbb{R}^d$. Introducing
\[M = M(\xi) := \e^{-\frac{|\xi|^2}{2}},\]
the VFK equation can be written as
\[
\partial_t f+ \xi\cdot \nabla_x f -\nabla V(x) \cdot \nabla_\xi f = \nabla_{\xi} \cdot (M  \nabla_{\xi} (M^{-1}f)).
\]
Let $W(\xi) = |\xi|^2/2$ and $\sigma = 1$. Then $M = \e^{-W/\sigma}$ and the above equation can be rewritten as
\[
\partial_t f+ \xi\cdot \nabla_x f-\nabla V(x) \cdot \nabla_\xi f = \sigma \nabla_{\xi} \cdot (\e^{-W/\sigma} \nabla_{\xi} (\e^{W/\sigma}f)),
\]
where the right-hand side has the same form as that of Eq.~\eqref{FP}.

It seems that we can also try the conservation form and heat equation form as for the Fokker-Planck equation, however, the third term $\nabla V(x) \cdot \nabla_\xi f$ on the left-hand side makes the direct use of the warped phase transformation no longer work. Let us consider the conservation form as an example.
Introduce the transformation $F(t,x,\xi,p) = \e^{-p}f(t,x,\xi)$ with periodic extension of the initial data to get
\[
\partial_t F + \xi\cdot \nabla_x F -\nabla V(x) \cdot \nabla_\xi f = \nabla_{\xi} \cdot (\e^{-W} \nabla_{\xi} (\e^{W} (- F_p))).
\]
Applying the discrete Fourier transformation on all variables, on easily obtains
\begin{align*}
&\xi\cdot \nabla_x F
 = \sum\limits_{l=1}^d \xi_l \partial_{x_l} F  \longrightarrow \i \sum\limits_{l=1}^d ( \bb{P}_l \otimes \bb{D}_l \otimes I) (\bb{F}_x \otimes \bb{F}_\xi \otimes \bb{F}_p), \\
&\nabla V(x) \cdot \nabla_\xi F
 = \sum\limits_{l=1}^d  \partial_{x_l} V \partial_{\xi_l} F  \longrightarrow \i \sum\limits_{l=1}^d ( \bb{V}_l \otimes \bb{P}_l \otimes I) (\bb{F}_x \otimes \bb{F}_\xi \otimes \bb{F}_p),
\end{align*}
where $\bb{V}_l = \sum_{\bb{i}} \partial_{x_l} V(x_{\bb{i}}) \ket{\bb{i}}$ is a diagonal matrix,
and
\begin{align*}
&  \nabla_{\xi} \cdot (\e^{-W} \nabla_{\xi} (\e^{W}(- F_p) )) \\
&  = \sum\limits_{l=1}^d \partial_{\xi_l}(\e^{-W} \partial_{\xi_l} (\e^{W} (- F_p)))
  = \i \sum\limits_{l=1}^d (-\i \partial_{\xi_l}) (\e^{-W} (-\i \partial_{\xi_l}) (\e^{W}  (-\i \partial_p F) )) \\
& \longrightarrow \i \sum\limits_{l=1}^d (I^{\otimes^d} \otimes \bb{P}_l \otimes I)(I^{\otimes^d} \otimes \bb{e}^{-\bb{W}_\xi}\otimes I)(I^{\otimes^d} \otimes \bb{P}_l \otimes I)(I^{\otimes^d} \otimes\bb{e}^{\bb{W}_\xi}\otimes I)(I^{\otimes^d} \otimes I^{\otimes^d} \otimes P_\mu)\bb{F} \\
& = \i \sum\limits_{l=1}^d(I^{\otimes^d} \otimes \bb{P}_l \bb{e}^{- \bb{W}_\xi} \bb{P}_l \bb{e}^{\bb{W}_\xi}  \otimes  P_\mu ) \bb{F},
\end{align*}
where $\bb{e}^{- \bb{W}_\xi}$ is a diagonal matrix with the diagonal vector given by
$\sum_{\bb{j}} W(\xi_{\bb{j}}) \ket{\bb{j}}$.
Let $\bb{B}_l = \bb{P}_l \bb{e}^{-\bb{W}_\xi} \bb{P}_l$, which are Hermitian matrices. One gets the following ODEs:
\[
\frac{\d }{\d t} \bb{F}  = \i \sum\limits_{l=1}^d ( -\bb{P}_l \otimes \bb{D}_l \otimes I + \bb{V}_l \otimes \bb{P}_l \otimes I + I^{\otimes^d} \otimes \bb{B}_l\bb{e}^{\bb{W}_\xi} \otimes P_\mu) \bb{F}.
\]

For this system, it doesn't work by introducing the new variables $\tilde{\bb{F}} = I^{\otimes^d} \otimes \bb{e}^{\bb{W}_\xi/2} \otimes F_p^{-1}$ since the first term on the right-hand side will change to $\bb{P}_l \otimes \bb{e}^{\bb{W}_\xi/2}\bb{D}_l\bb{e}^{-\bb{W}_\xi/2} \otimes I$, which are not Hermitian matrices.

To resolve this problem, one can first derive an ODE system resulting from the discretisation of $x$ and $\xi$ variables, and then apply the generalised approach in Subsect.~\ref{subsect:generalisation}, with the details omitted.

\subsection{The Liouville representation  for nonlinear ODEs}

Consider the following nonlinear ODEs
\begin{equation}\label{ODEsFx}
\frac{\d q(t)}{\d t} = F(q(t)), \quad q(0) = q_0 , \quad q = [q_1,\cdots,q_d]^T.
\end{equation}
For $x = (x_1,\cdots,x_d)$, let $\delta(x) = \Pi_{i=1}^d \delta(x_i)$  be the Dirac delta distribution. The Liouville equation corresponding to \eqref{ODEsFx} can be derived by considering
 a function $\rho(t, x): \mathbb{R}^+ \times \mathbb{R}^d \to \mathbb{R}$, defined by
\begin{equation*}%\label{rhodelta}
\rho(t,x) = \delta( x-q(t)),
\end{equation*}
which represents the probability distribution in space $x$ that corresponds to the solution $x =q$. By the properties of the delta function, one obtains the solution of \eqref{ODEsFx} by taking the moment:
\begin{equation}\label{xrho}
q(t) = \int x\delta( x-q(t)) \d x = \int x\rho(t,x) \d x.
\end{equation}
To this end, we can characterize the dynamics of $\rho(t,x)$ and find the solution $q(t)$ via \eqref{xrho}.

One can check that $\rho$ satisfies, in the weak sense, the linear $(d + 1)$-dimensional PDE
\begin{equation*}
\begin{cases}
\partial_t \rho(t,x) + \nabla_x \cdot [F(x) \rho(t,x)  ] = 0,\\
\rho_0(x):=\rho(0,x) = \delta(x-q_0).
\end{cases}
\end{equation*}
Since the initial data involves a delta function, we consider the following problem with the smoothed initial data \cite{JinLiuYu2022nonlinear}
\begin{equation*}
\begin{cases}
\partial_ t u(t,x) + \nabla \cdot [F(x) u(t,x)  ] = 0,\\
u(0,x) = u_0(x):= \delta_\omega(x-q_0).
\end{cases}
\end{equation*}

Let $u_1 = F_i(x) u(t,x)$. According to the notations for the Fourier spectral method, one has
\[ -\i \partial_{x_i} u_1 \longrightarrow  \hat{P}_i^{\rm d}\bb{u}_1 = \bb{P}_i \bb{u}_1 = \bb{P}_i \Lambda_{F_i} \bb{u},\]
where $\Lambda_{F_i} = \text{diag}(\bb{F}_i)$ is a diagonal matrix and $\bb{F}_i= \sum_{\bb{j}} F_i(x_{\bb{j}}) \ket{\bb{j}}$. The resulting system of ordinary differential equations is
\begin{equation}\label{SpectralSystemLiouvilleRep}
\begin{cases}
\frac{\d }{\d t} \bb{u}(t) = -\i A \bb{u}(t),\\
\bb{u}(0) = \bb{u}^0 = ( u_0(x_{\bb{j}}) ),
\end{cases}
\end{equation}
where $A = \sum_{i=1}^d A_i$ with $A_i = \bb{P}_i \Lambda_{F_i}$. In Ref.~\cite{JinLiuYu2022nonlinear} we proposed a quantum simulation method for the above problem by using the time splitting approach. However, the simulation protocol there is different from the traditional time-marching Hamiltonian simulation since non-unitary procedures are involved at each time step, leading to exponential increase of the cost. By the generalised framework in Subsect.~\ref{subsect:generalisation}, we are ready to recover a true Hamiltonian simulation.

Let us still consider the splitting in \cite{JinLiuYu2022nonlinear}. In fact, there is no need to apply the splitting since each $\bb{A}_i$ has the similar structure. Here we just indicate that our protocol also works for the splitting. The evolution of \eqref{SpectralSystemLiouvilleRep} can be written as
\[\bb{w}(t+\Delta t) = \e^{-\i (A_1 + \cdots + A_d) \Delta t } \bb{w}(t),\]
in which the evolutionary operator can be approximated by the first-order product formula
\begin{equation}\label{Wdelta}
U_{\Delta t} = \e^{-\i A_d \Delta t}\cdots \e^{-\i A_1 \Delta t} .
\end{equation}
Then the problem is reduced to the simulation of each $A_j$, where $A_j$ is not necessarily symmetric. Consider the decomposition $-\i A_j  = H_1^j +  \i H_2^j$, where $H_1^j$  and $H_2^j$ are Hermitian matrices. Following the discussion in Subsect.~\ref{subsect:generalisation}, one can construct a Hamiltonian system associated with the evolutionary operator $\e^{-\i A_j \Delta t} $:
\[ \frac{\d}{\d t} \bb{w}(t) = \i ( -H_1^j \otimes P_\mu  +  H_2^j \otimes I ) \bb{w}(t), \]
where $\bb{w}$ encodes all the grid values of the variable $\bb{v}$ defined by \eqref{u2v} (note that $\bb{v}$ corresponds to the original system \eqref{SpectralSystemLiouvilleRep}).
In contrast to the change of variables in \cite{JinLiuYu2022nonlinear}, the transformations here are the same in every time step. We therefore restore the time-marching Hamiltonian simulation for the new variables.

\section{Summary}

We provided technical details for a new method, referred to as Schr\"odingerisation, for solving general linear PDEs using quantum simulation. The idea is to introduce a simple warped phase transformation that can translate the PDEs into a `Schr\"odingerised' or Hamiltonian system, without employing more sophisticated methods.  This enables quantum simulation for these PDEs (and also ODEs).

This approach was applied to several typical examples, including the heat, convection, Fokker-Planck, linear Boltzmann and Black-Scholes equations. It can be extended to Schr\"odingerise general linear partial differential equations, including the Vlasov-Fokker-Planck equation and the Liouville representation equation for nonlinear ordinary differential equations. It also has the potential to find a variety of applications in time-dependent or independent boundary value problems.

\section*{Acknowledgements}
SJ was partially supported by the NSFC grant No.~12031013, the Shanghai Municipal Science and Technology Major Project (2021SHZDZX0102), and the Innovation Program of Shanghai Municipal Education Commission (No. 2021-01-07-00-02-E00087).  NL acknowledges funding from the Science and Technology Program of Shanghai, China (21JC1402900). YY was partially supported by China Postdoctoral Science Foundation (no. 2022M712080).

%\addcontentsline{toc}{section}{References}
\bibliographystyle{plain} %plain, unsrt, alpha
\bibliography{QnovelSim}
%\bibliography{Nonlinear}

\end{document}